\documentclass
[copyright,creativecommons, UKenglish]{eptcs}

\usepackage[nomargin,inline,index,%
]{fixme}
\fxusetheme{colorsig}
\FXRegisterAuthor{fxJP}{anfxJP}{\underline{\textbf{JP}}}
\FXRegisterAuthor{fxMD}{anfxMD}{}
\FXRegisterAuthor{fxNY}{anfxNY}{\underline{\textbf{NY}}}
\FXRegisterAuthor{fxSG}{anfxSG}{\underline{\textbf{SG}}}
\FXRegisterAuthor{fxSJ}{anfxSJ}{\underline{\textbf{SJ}}}

\usepackage{makeidx}
\usepackage{color}
\usepackage{stmaryrd}
\usepackage{mathpartir,amsmath,amssymb
}
\usepackage{prooftree}
\usepackage{mathtools}
\newtheorem{theorem}{Theorem}[section]
\newtheorem{lemma}[theorem]{Lemma}
\newtheorem{example}[theorem]{Example}
\newtheorem{definition}[theorem]{Definition}
\newtheorem{proposition}[theorem]{Proposition}
\newenvironment{proof}{{\em Proof.}}{\medskip}

\newcommand{\RedG}{\Longrightarrow}
\newcommand{\dlsqb}{[\![}
\newcommand{\drsqb}{]\!]}

\newcommand{\T}{{\mathsf T}}
\newcommand{\R}{R}
\renewcommand{\S}{S}

\newcommand{\rulename}[1]{\text{\small[\textsc{#1}]}}

\newcommand{\pp}{{\sf p}}
\newcommand{\q}{\pq}
\newcommand{\pq}{{\sf q}}
\newcommand{\pr}{{\sf r}}
\newcommand{\ps}{{\sf s}}
\newcommand{\e}{\kf{e}}
\newcommand{\x}{x}

\newcommand{\val}{\kf{v}}
\newcommand{\valn}{\kf{n}}
\newcommand{\valr}{\kf{i}}

\newcommand{\sep}{\ensuremath{~~\mathbf{|\!\!|}~~ }}
\newcommand{\kf}[1]{\ensuremath{\mathsf{#1}}}
\newcommand{\pc}{\ensuremath{~|~}}

\newcommand{\G}{\ensuremath{{\sf G}}}
\newcommand{\Gvti}[5]{\ensuremath{#1\to#2:\{#3_i({#4}_i). #5_i \}_{i \in I}}}
\newcommand{\ty}{\textbf{t}}

\newcommand{\N}{\ensuremath{\mathcal M}}
\newcommand{\M}{\ensuremath{\mathcal M}}

\newcommand{\pa}[2]{#1 \triangleleft  #2}
\newcommand{\set}[1]{\{#1\}}
\newcommand{\eval}[2]{#1 \downarrow #2}
\newcommand{\true}{\kf{true}}
\newcommand{\false}{\kf{false}}

\newcommand{\myrule}[3]{\begin{prooftree} #1 \justifies   #2   \using{\rln{#3}} \end{prooftree}}
\newcommand{\rln}[1]{\textsc{#1}}
\newcommand{\der}[3]{ #1 \vdash   #2  :#3}

\newcommand{\CP}[1]{ {\mathcal P}(#1)}
\newcommand{\CG}[2]{ {\mathcal G}(#1,#2)}
\newcommand{\CGZ}[3]{ {\mathcal G}_0(#1,#2,#3)}

\newcommand{\CC}[1]{ {\mathcal C}[#1]}

\newcommand{\participant}[1]{\mathtt{pt}\{#1\}}

\newcommand{\proj}[2]{ #1 \upharpoonright #2}
\newcommand{\sub}[2]{\set{#1/#2}}

\newcommand{\procin}[3]{#1 ?  #2 .#3}
\newcommand{\procout}[3]{#1 ! #2. #3}

\newcommand{\PP}{\ensuremath{P}}
\newcommand{\Q}{\ensuremath{Q}}
\newcommand{\cond}[3]{\kf{if}~ #1 ~\kf{then} ~#2 ~\kf{else}~#3}
\newcommand{\inact}{\ensuremath{\mathbf{0}}}

\newcommand{\external}{+}

\newcommand{\tend}{\mathtt{end}}
\newcommand{\tbool}{\mathtt{bool}}
\newcommand{\tnat}{\mathtt{nat}}
\newcommand{\tint}{\mathtt{int}}
\newcommand{\tin}[3]{#1?#2(#3)}
\newcommand{\tout}[3]{#1!#2(#3)}
\newcommand{\tdag}[3]{#1\dagger #2(#3)}
\newcommand{\tddag}[3]{#1\ddagger #2 (#3)}
\newcommand{\tinternal}{\vee}
\newcommand{\texternal}{\wedge}

\renewcommand{\S}{S}

\newcommand{\lts}[1]{\xrightarrow{#1}}
\newcommand{\redG}[4]{#1\setminus#2\xrightarrow#3#4}
\newcommand{\SortI}[2]{ {#1}_{#2}}

\newcommand{\Econtext}{\mathcal{E}}

\newcommand{\subt}{\leqslant}
\newcommand{\subs}{\leq\vcentcolon}
\newcommand{\red}{\longrightarrow}
\newcommand{\nsubt}{\not\trianglelefteq}

\newcommand{\fsqrt}[1]{{\tt neg}(#1)}
\newcommand{\fneg}{\fsqrt}
\newcommand{\fsucc}[1]{{\tt succ}(#1)}

\newcommand{\stuck}[1]{{\tt stuck}(#1)}
\newcommand{\dual}[1]{\overline #1}

\newcommand{\cinferrule}[3][]{
  \mprset{fraction={===},
  fractionaboveskip=0.2ex,
  fractionbelowskip=0.4ex}
  \inferrule[#1]{#2}{#3}
}

\newcommand{\I}{\bigwedge\!\!\!\!\bigwedge}

\definecolor{ceca}{rgb}{1,0.5,0}

\newcommand{\SUB}[3]{{\mathcal S}(#1,#2,#3)}
\newcommand{\Gvtir}[5]{\ensuremath{#1\to#2:\{#3_i({#4}_i). \redG{#5_i}\pp\ell\q \}_{i \in I}}}

 \newcommand{\myformula}[1]{\\[5pt]\centerline{$#1$}\\[5pt]}
  \newcommand{\myformulaA}[1]{\centerline{$#1$}}
  \newcommand{\myformulaC}[1]{\\[3pt]\centerline{$#1$}}
 \newenvironment{mytable}[1]
               {\begin{table} [#1]}{\vspace{-3pt}
               \end{table}}
                
  \newcommand{\myparagraph}[1]{\vspace{-25pt}
  \paragraph{#1}}

   \newcommand{\mysection}[1]{\vspace{-7pt}
  \section{#1}\vspace{-7pt}}
  
   \newenvironment{mydefinition}
               {\begin{definition}\vspace{-4pt}
               }{\vspace{-2pt}
               \end{definition}}
               
               \newenvironment{mydefinitionA}
               {\begin{definition}
               }{
               \end{definition}}
               
                 \newenvironment{mytheorem}[1]
               {\begin{theorem}{\bf{(#1)}}\vspace{-2pt}
               }{\vspace{-3pt}
               \end{theorem}}
               
               \newenvironment{mytheoremA}[1]
               {\begin{theorem}{\bf{(#1)}}
               }{\vspace{-3pt}
               \end{theorem}}
               
               \newenvironment{mylemma}[1]
               {\begin{lemma}{\bf{(#1)}}\vspace{-2pt}
               }{\vspace{-3pt}
               \end{lemma}}
               
                \newenvironment{mylemmaB}[1]
               {\begin{lemma}{\bf{(#1)}}\vspace{-2pt}
               }{
               \end{lemma}}

                \newenvironment{mylemmaA}
               {\begin{lemma}\vspace{-1pt}
               }{\vspace{-3pt}
               \end{lemma}}
               
               \newenvironment{myitemize}
               {\begin{itemize}\vspace{-5pt}
               \topsep0pt\parskip0pt\partopsep0pt\itemsep0pt\leftmargin0pt\itemsep2pt\labelwidth0pt\labelsep3pt
              }
               {\vspace{-2pt}
               \end{itemize}}
               
                \newenvironment{myenumerate}
               {\begin{enumerate}
               \topsep0pt\parskip0pt\partopsep0pt\itemsep0pt\leftmargin0pt\itemsep2pt\labelwidth0pt\labelsep3pt
              }
               {\vspace{-2pt}
               \end{enumerate}}
               
               \newenvironment{myenumerateB}
               {\begin{enumerate}
               \topsep0pt\parskip0pt\partopsep0pt\itemsep0pt\leftmargin0pt\itemsep2pt\labelwidth0pt\labelsep3pt
              }
               {
               \end{enumerate}}
               
               \newenvironment{myexample}
               {\begin{example}\vspace{-2pt}
               }{\vspace{-2pt}
               \end{example}}
               \newenvironment{myexampleA}
               {\begin{example}\vspace{-2pt}
               }{\vspace{-20pt}
               \end{example}}

                \newenvironment{myexampleC}
               {\begin{example}
               }{\vspace{-2pt}
               \end{example}}

\begin{document}

%
\title{Precise subtyping for synchronous multiparty sessions
\thanks{Partly supported by COST IC1201 BETTY and DART bilateral 
project between Italy and Serbia.}}

\def\titlerunning{Precise subtyping for synchronous multiparty sessions}

\author{
  Mariangiola Dezani-Ciancaglini\institute{
  Universit\`a di Torino, Italy}\thanks{Partly supported by MIUR PRIN Project CINA Prot. 2010LHT4KM and Torino University/Compagnia San Paolo Project SALT.}
  \and
 Silvia Ghilezan\institute{
  Univerzitet u Novom Sadu, Serbia}
  \and 
  Svetlana Jak\v{s}i\'c\institute{
  Univerzitet u Novom Sadu, Serbia}
   \and
  Jovanka Pantovi\'c\institute{
  Univerzitet u Novom Sadu, Serbia}
  \and 
  Nobuko Yoshida\institute
{Imperial College London}\thanks{Partly supported by 
EPSRC EP/K011715/1, EP/K034413/1, and 
EP/L00058X/1, and EU Project FP7-612985 UpScale. 
}
}

%

\def\authorrunning{Dezani, Ghilezan, Jak\v{s}i\'c, Pantovi\'c, Yoshida}

\providecommand{\event}{PLACES 2015}

\maketitle

\begin{abstract}
  The notion of subtyping has gained an important role both in
theoretical and applicative domains: in lambda and concurrent calculi
as well as in programming languages. The soundness and the
completeness, together referred to as the preciseness of subtyping,
can be considered from two different points of view: 
operational and denotational. The former preciseness has been recently developed with
respect to type safety, i.e. the safe replacement of a term of a
smaller type when a term of a bigger type is expected. The latter preciseness is based on the denotation of a
type which is a mathematical object that describes the meaning of the
type in accordance with the denotations of other expressions from the
language. The result of
this paper is the operational and denotational preciseness of the
subtyping for a synchronous multiparty session calculus. The novelty
of this paper is the introduction of characteristic global types to
prove the operational completeness.
\end{abstract}

\mysection{Introduction}
In modelling distributed systems, where many processes interact by means of
message passing, one soon realises that most interactions are meant to
occur within the scope of private channels according to disciplined
protocols.
Following~\cite{HYC08}, we call such private interactions \emph{multiparty sessions} and the protocols
that describe them \emph{multiparty session types}.

The ability to describe complex interaction protocols by means of a
formal, simple and yet expressive type language can have a profound
impact on the way distributed systems are designed and developed. This
is witnessed by the fact that some important standardisation bodies
for web-based business and finance protocols \cite{CDL,UNIFI,savara}
have recently investigated design and implementation frameworks for
specifying message exchange rules and validating business logic based
on the notion of multiparty sessions, where multiparty session
types are ``shared agreements'' between teams of programmers
developing possibly large and complex distributed protocols or
software systems.

\emph{Subtyping} has been extensively studied 
as one of the most interesting issues in type theory. 
The correctness of subtyping relations has been usually provided as 
the operational soundness: If $\T$ is a subtype of $\T'$ (notation $\T\leq\T'$), then a term of type $\T$ may be provided whenever a term of type $\T'$ is needed, see~\cite{pier02} (Chapter 15) and~\cite{harp13} (Chapter 23).
The converse direction, the operational completeness,  
has been largely ignored 
in spite of its usefulness to 
define the greatest subtyping relation ensuring type safety. If $\dlsqb \T  \drsqb$ is the set interpretating type $\T$, then a subtyping is denotationally sound when $\T \leq \T' $ implies  $\dlsqb \T  \drsqb \subseteq  \dlsqb \T'  \drsqb$ and denotationally complete when $ \dlsqb \T  \drsqb \subseteq  \dlsqb \T'  \drsqb  $ implies  $ \T \leq \T'$. \emph{Preciseness} means both soundness and completeness.

Operational
preciseness has been first introduced in~\cite{BHLN12} for a
call-by-value  $\lambda$-calculus with sum, product and recursive
types. Both operational and denotational preciseness have been studied
in~\cite{DG14} for a  $\lambda$-calculus with choice and parallel
constructors~\cite{SIAM} and in~\cite{cdy14} for binary sessions~\cite{THK}. 

These facts ask for investigating precise subtyping for 
multiparty session types, the subject of this paper.
Subtyping for session calculi can be defined to assure safety of substitutability of either channels~\cite{GH05} or processes~\cite{DemangeonH11}. We claim that 
substitutability of processes better fits the notion of preciseness.

We show the operational and denotational preciseness of the
subtyping introduced in~\cite{DemangeonH11} for a simplification of the synchronous multiparty
session calculus in~\cite{KY13}. For the operational preciseness we take the view that well-typed sessions never get stuck. For the denotational preciseness we interpret a type as the set of processes having that type. 
 
The most technical challenge is the operational completeness, 
which requires a non trivial extension of the method used in the case of binary sessions. The core of this extension is 
the construction of {\em characteristic global types}.  
\paragraph{Outline} The calculus and its type system are introduced in Sections~\ref{sec:msc} and \ref{sec:ts}, respectively. Sections~\ref{sec:op} and~\ref{sec:denotation} contain the proofs of operational and denotational preciseness. Section~\ref{ex} illustrates the operational preciseness by means of an example. Some concluding remarks are the content of Section~\ref{conc}.

\mysection{Synchronous Multiparty Session Calculus}\label{sec:msc}

This section introduces syntax and semantics of a synchronous
multiparty session calculus.
Since our focus is on
subtyping, we simplify the calculus in \cite{KY13} 
eliminating both shared channels
for session initiations and session channels for communications inside
sessions. We conjecture the preciseness of the subtyping in~\cite{DemangeonH11} also for the full calculus, but we could not use the present approach for the proof, since well-typed interleaved sessions can be stuck~\cite{CDYP15}.

\noindent
\myparagraph{Syntax}
A \emph{multiparty session} is a series of 
interactions between a fixed number of
participants, possibly with branching and recursion, and serves as a unit
of abstraction for describing communication protocols.

We use the following base sets:  \emph{values}, ranged over by $\val,\val',\ldots$;
 \emph{expressions}, ranged over by $\e,\e',\ldots$;
\emph{expression variables}, ranged over by
$x,y,z\dots$; \emph{labels}, ranged over by $\ell,\ell',\dots$;
 \emph{session participants}, ranged over by $\pp,\pq,\ldots$;
\emph{process variables}, ranged over by $X,Y,\dots$; 
\emph{processes}, ranged over by $P,Q,\dots$; and \emph{multiparty sessions}, ranged over by $\N,\N',\dots$.

The values are natural numbers $\valn$, integers $\valr$, and boolean values $\true$ and $\false$.  The expressions $\e$ are variables or values or expressions built from expressions by applying the operators ${\tt succ}, {\tt neg}, \neg, \oplus,$ or the relation $>.$  An \emph{evaluation context} $\Econtext$ is an expression with exactly one hole, built in the same manner from expressions and the hole.

Processes $\PP$ are defined by:
\begin{myformula} {\begin{array}{lll}\PP   &   ::=  &  \procin{\pp}{\ell(\x)}{\PP}      \sep     \procout \pp  {\ell(\e)} \PP    
            \sep     \PP + \PP                                    
            \sep     \cond{\e} \PP  \PP                                     \sep     \mu X.\PP                                                \sep     X                                           
                   \sep     \inact  \end{array}}\end{myformula}
\indent    
The input process $\procin{\pp}{\ell(\x)}{\PP}$ waits for an expression with label $\ell$ from participant $\pp$ and
the output process $\procout{\q}{\ell(\e)}{\Q}$ sends the value of expression $\e$ with label $\ell$  to participant $\q$. 
The external choice $P+ Q$ offers to choose either $P$ or $Q$. The process
$\mu X.\PP$ is a recursive process. 
We take an equi-recursive view, not distinguishing
between a process $\mu X.\PP$ 
and its unfolding
$\PP\sub{\mu X.\PP}{X}$. 
 We assume that the recursive processes are guarded, i.e. $\mu X.X$ is not a process. 

A multiparty session $\N$ is a parallel composition of pairs (denoted by $\pa\pp\PP$) of participants and processes: 
\myformula{\begin{array}{lll}\N  &  ::=   & \pa\pp\PP  \sep \N \pc\N                       \end{array}}
We will use 
   $\sum\limits_{i\in I}  \PP_i$  as short for $\PP_1+\ldots+\PP_n,$ and
   $\prod\limits_{i\in I} \pa{\pp_i}{\PP_i}$ as short  for $\pa{\pp_1}{\PP_1}\pc \ldots\pc \pa{\pp_n}{\PP_n},$ 
where $I=\set{1,\ldots,n}$.

If $\pa\pp\PP$ is well typed (see Table~\ref{tab:sync:typing}), then participant $\pp$ does not occur in process $\PP$, since we do not allow self-communications.

\noindent \myparagraph{Operational semantics}
{\em The value $\val$ of expression $\e$} (notation $\eval\e\val$) is as expected, see Table~\ref{tab:evaluation}. The successor operation ${\tt succ}$ is defined only on  natural numbers, the negation ${\tt neg}$ is defined on integers (and then also on natural numbers), and $\neg$ is defined only on boolean values. The internal choice $\e_1\oplus\e_2$ evaluates either to the value of $\e_1$ or to the value of $\e_2$. 
\begin{mytable}{}
\centerline{$\begin{array}[t]{@{}c@{}}
\eval{\fsucc\valn}(\valn +1)
\quad
\eval{\fsqrt\valr}(-\valr)
\quad
\eval{\neg\true}\false \quad \eval{\neg\false}\true 
\quad \eval\val\val\\
\eval{(\valr_1>\valr_2)}{\begin{cases}
 \true     & \text{if }\valr_1>\valr_2, \\
  \false    & \text{otherwise}
\end{cases}} \qquad
\inferrule[]{\eval{\e_1}\val\text{ or }\eval{\e_2}\val}
{\eval{\e_1\oplus\e_2}{\val}}
\qquad
\inferrule[]{\eval{\e}{\val}\quad\eval{\Econtext(\val)}{\val'}}{\eval{\Econtext(\e)}{\val'}}
\end{array}
$}
\caption{\label{tab:evaluation}  Expression evaluation.}
\end{mytable}
 \begin{mytable}{}
\centerline{$
\begin{array}[t]{@{}c@{}}
\inferrule[\rulename{s-extch 1}]{}{
   \PP \external \Q \equiv \Q \external \PP
      }
\qquad
\inferrule[\rulename{s-extch 2}]{}{
   (\PP \external \Q ) \external \R\equiv \PP \external (\Q \external \R)
      }
\qquad
\inferrule[\rulename{s-multi}]{}{
   \PP \equiv \Q  \Rightarrow \pa\pp\PP\equiv\pa\pp\Q
      }
\\\\\      
  \inferrule[\rulename{s-par 1}]{}{
    \pa\pp{\inact} \pc \N \equiv \N
      }
\qquad
  \inferrule[\rulename{s-par 2}]{}{
    \N \pc \N' \equiv \N' \pc \N
      }
 \qquad
 \inferrule[\rulename{s-par 3}]{}{
    (\N \pc \N') \pc \N'' \equiv \N \pc( \N' \pc \N'')
      }
\end{array}
$}
\caption{Structural congruence.}
\label{tab:sync:congr}
\end{mytable}
\begin{mytable}{}
\myformulaA{
\begin{array}[ht]{@{}c@{}}
\inferrule[\rulename{r-comm}]{
     j \in I \qquad \eval{\e}{\val}}{
    \pa\pp{\sum\limits_{i\in I} \procin{\q}{\ell_i(\x)}{\PP_i}}\; \pc \;\pa\q{\procout \pp {\ell_j(\e)} \Q}
    \red
    \pa\pp{\PP_j}\sub{\val}{\x}\;\pc\;\pa\q\Q
    }
  \qquad
  \inferrule[\rulename{t-conditional}]{
   \eval{\e}{\true}}{
    \pa\pp{\cond{\e}{\PP}{\Q}}  \red \pa\pp\PP
   }
 \\ \\
  \inferrule[\rulename{f-conditional}]{
   \eval{\e}{\false}}{
    \pa\pp{\cond{\e}{\PP}{\Q}}  \red \pa\pp\Q
   }
  \qquad
 \inferrule[\rulename{r-context}]{
    \N \red \N'
  }
  { \CC{\N} \red \CC{\N'}
  }   
  \qquad
  \inferrule[\rulename{r-struct}]{
  \N'_1\equiv \N_1 \quad \N_1\red \N_2 \quad \N_2 \equiv \N'_2
  }
  { 
   \N'_1 \red \N'_2
  }
\end{array}
}
\caption{Reduction rules.}
\label{tab:sync:red}
\end{mytable}

The {\em computational rules of multiparty sessions} (Table~\ref{tab:sync:red}) are closed with respect to the structural congruence defined in Table~\ref{tab:sync:congr}  and the following reduction contexts:
\begin{myformula} {
  \CC{\cdot} ::= [\cdot] \sep \CC{\cdot}\pc\N  
}\end{myformula}
In rule \rulename{r-comm} participant $\q$ sends the value $\val$ choosing label $\ell_j$ to participant $\pp$ which offers inputs on all labels $\ell_i$ with $i\in I$. 
We use $\red^*$ with the standard meaning.

 In order to define the operational preciseness of subtyping it is crucial to formalise when a  multiparty session contains communications that will never be executed. 

\begin{mydefinition}
  A multiparty session $\M$ is {\em stuck} if $\M\not\equiv\pa\pp\inact$ and there is no multiparty session $\M'$ such that  $\M\red\M'.$  A multiparty session $\M$ gets {\em stuck}, notation $\stuck\M,$ if it reduces to a stuck multiparty session.
\end{mydefinition}

\mysection{Type System}\label{sec:ts}

This section introduces the type system, which is a simplification of that in~\cite{KY13} due to the new formulation of the calculus. 

\noindent \myparagraph{Types} {\em Sorts}  are ranged over by $\S$ and defined by:\qquad
$ \S    \quad   ::=   \quad                          \tnat \sep \tint \sep\tbool$

\noindent
{\em Global types} generated by:
\begin{myformula}{
\begin{tabular}{rclrrrclr}
 $\G$ & $::=$ & $\Gvti \pp\q \ell \S {\G}$  & \sep &  $\mu\ty. \G$ & \sep &$\ty$   &\sep &$\tend$
\end{tabular}}\end{myformula}
describe the whole conversation scenarios of multiparty sessions. {\em Session types} correspond to projections of global types on the individual participants. 
Inspired by~\cite{P11}, we use intersection and union types instead of standard branching and selection~\cite{HYC08} to take advantage from the subtyping induced by subset inclusion.
The grammar of session types, ranged over by $\T$, is then 
\begin{myformula}{
\begin{tabular}{rclrrrrrclr}
 $\T$ & $::=$ & $\bigwedge_{i\in I}\tin\pp{\ell_i}{\S_i}.\T_i$  &\sep & $\bigvee_{i\in I}\tout\q{\ell_i}{\S_i}.\T_i$ & \sep & $\mu\ty. \T$& \sep & $\ty$  &\sep &$\tend$
\end{tabular}}\end{myformula}
We require  that $\ell_i\not=\ell_j$ with $i\not=j$ and $i,j\in I$  and recursion to be guarded in both global and session types. Recursive types with the same regular tree are considered equal~\cite[Chapter 20, Section 2]{pier02}. In writing types we omit unnecessary brackets, intersections, unions and $\tend$. 

We extend the original definition of projection of global types onto participants~\cite{HYC08} in the line of~\cite{YDBH10}, but keeping the definition simpler than that of~\cite{YDBH10}. This generalisation is enough to project the characteristic global types of next Section. We use the partial operator $\I$ on session types. This operator applied to two identical  types gives one of them, applied to two intersection types with same sender and different labels gives their intersection and it is undefined otherwise, see Table~\ref{pro}. The same table gives 
the {\em projection} of the global type $\G$ onto the participant $\pr$, notation $\proj \G \pr$. This projection allows participants to receive different messages in different branches of global types. 
\begin{mytable}{} 
\centerline{$\begin{array}{c}\T\I\T'=\begin{cases}
 \T     & \text{if }\T=\T', \\
 \T\wedge\T'     & \text{if }\T=\bigwedge_{i\in I}\tin\pp{\ell_i}{\S_i}.\T_i \text{ and } \T'=\bigwedge_{j\in J}\tin\pp{\ell'_j}{\S'_j}.\T'_j \\
 &\text{ and } \ell_i\not=\ell_j'\text{ for all }i\in I, j\in J \\
 \text{undefined}     & \text{otherwise}.
\end{cases}\\[9mm]
\proj {\Gvti \pp\q \ell \S {\G}} \pr=\begin{cases}
\bigvee_{i\in I}\tout\q{\ell_i}{\S_i}. \proj{\G_i}\pr     & \text{if }\pr=\pp, \\
\bigwedge_{i\in I}\tin\pp{\ell_i}{\S_i}. \proj{\G_i}\pr     & \text{if }\pr=\q, \\
   \I_{i\in I}\proj{\G_i}\pr    & \text{if $\pr\not=\pp$, $\pr\not=\q$ and $\I_{i\in I}\proj{\G_i}\pr$ is defined}.
\end{cases}\\[7mm]
\proj {(\mu\ty.\G)}\pr=\begin{cases}
 \mu\ty.\proj{\G}\pr     & \text{if $\pr$ occurs in $\G$}, \\
 \tend     & \text{otherwise}.
\end{cases}\qquad\qquad \proj{\ty}\pr=\ty\qquad\qquad \proj{\tend}\pr=\tend
\end{array}
$}
\caption{Projection of global types onto participants.}
\label{pro}
\end{mytable}
\begin{myexampleA}
If $\G=\pp\to\q:\set{\ell_1(\tnat).\G_1, \ell_2(\tbool).\G_2}$, where\\ $\G_1=\q\to\pr:\ell_3(\tint)$ and $\G_2=\q\to\pr:\ell_5(\tnat)$ and $\pr\neq \pp$,  then
\myformula{\proj{\G}\pr=\proj{\G_1}\pr\I\proj{\G_2}\pr=\tin\q{\ell_3}{\tint}\I\tin\q{\ell_5}{\tnat}=\tin\q{\ell_3}{\tint}\wedge\tin\q{\ell_5}{\tnat}.}
\end{myexampleA} 
\noindent \myparagraph{Subtyping}
\begin{mytable}{h}
\centerline{$
\begin{array}{@{}c@{}}
\inferrule[\rulename{sub-end}]{}
  {\tend \subt \tend}\qquad
\cinferrule[\rulename{sub-in}]{
\forall i\in I: \quad \S_i'\subs \S_i \quad \T_i \subt\T_i' }{
  \bigwedge_{i\in I\cup J} \tin\pp{\ell_i}{\S_i}.\T_i \subt \bigwedge_{i\in I}\tin\pp{\ell_i}{\S_i'}.\T_i'
}
\qquad
\cinferrule[\rulename{sub-out}]{
 \forall i\in I: \quad \S_i \subs \S'_i  \quad  \T_i \subt \T'_i  
}{\bigvee_{i\in I}
  \tout\pp{\ell_i}{\S_i}.\T_i
  \subt\bigvee_{i\in I\cup J}
 \tout\pp{\ell_i}{\S'_i}.\T_i'
 }
\end{array}
$}
\caption{\label{tab:sync:subt} Subtyping rules.}
\end{mytable}
\begin{mytable}{}
\centerline{
$\SUB\Theta\T{\T'}=\left\{
\begin{array}{ll}
    \true                                                                                    & \!\!\!\!\!\! \text{ if }  \T\subt\T'\in\Theta \text{ or }  \T=\T'  \\
   \&_{i\in I}\SUB{\Theta\cup\set{\T\subt\T'}}{\T_i}{\T'_i}  & \!\!\!\!\!\!\text{ if } (\T=\bigwedge\limits_{i\in I\cup J} \tin\pp{\ell_i}{\S_i}.\T_i \text{ and }\T'=\bigwedge\limits_{i\in I}\tin\pp{\ell_i}{\S_i'}.\T_i'  \\
                                                                                                & \qquad \text{ and } \forall i\in I: \S_i'\subs \S_i) \text{ or }  \\
                                                                                                & \!\!\!\!\!\phantom{ \text{ if }} (\T=\bigvee\limits_{i\in I} \tout\pp{\ell_i}{\S_i}.\T_i \text{ and }\T'=\bigvee\limits_{i\in I\cup J}\tout\pp{\ell_i}{\S_i'}.\T_i' \\
                                                                                                 & \qquad \text{ and }\forall i\in I: \S_i\subs \S_i' )\\
  \false                                                                                     & \!\!\!\!\!\!\text{ otherwise }
\end{array}
\right.
$}
\caption{\label{ds} The procedure $\SUB\Theta\T{\T'}$.}
\end{mytable}
{\em Subsorting }$\subs$ on sorts is the minimal reflexive and transitive closure of the relation induced by the rule:
  $\tnat \subs \tint$.
{\em Subtyping} 
$\subt$ on session types takes into account the contra-variance of inputs, the covariance of outputs, and the standard rules for intersection and union.  Table~\ref{tab:sync:subt} gives the subtyping rules: the double line in rules indicates
that the rules are interpreted {\em coinductively}~\cite{pier02} (Chapter 21). Subtyping can be easily decided, see for example~\cite{GH05}. For reader convenience Table~\ref{ds} gives the procedure $\SUB\Theta\T{\T'}$, where $\Theta$ is a set of subtyping judgments. This procedure terminates since unfolding of session types generates regular trees, so $\Theta$ cannot  grow indefinitely and we have only a finite number of subtyping judgments to consider. 
Clearly $\SUB\emptyset\T{\T'}$ is equivalent to $\T\subt\T'$.

\noindent \myparagraph{Typing system}
We distinguish three kinds of typing judgments
\begin{myformula}{
  \der\Gamma\e\S \qquad\qquad \der\Gamma\PP\T \qquad\qquad \der{}\N\G,
}\end{myformula}
where $\Gamma$ is the environment $\Gamma ::= \emptyset \sep \Gamma, x:\S \sep \Gamma, X:\T$ that associates expression variables with sorts and process variables with session types. The typing rules for expressions are standard, see Table~\ref{tab:te}. 
\begin{mytable}{t}
\centerline{$
\begin{array}{c}
  \der{\Gamma}{\valn}{\tnat} \qquad \der{\Gamma}{\valr}{\tint}
  \qquad
  \der{\Gamma}{\true}{\tbool} \qquad \der{\Gamma}{\false}{\tbool}
  \qquad
  \der{\Gamma,x:\S}{x}{\S}  
  \\ \\
  \myrule{\der{\Gamma}{\e}{\tnat}}{\der{\Gamma}{\fsucc\e}{\tnat}}{}
  \qquad
  \myrule{\der{\Gamma}{\e}{\tint}}{\der{\Gamma}{\fsqrt\e}{\tint}}{}
  \qquad
  \myrule{\der{\Gamma}{\e}{\tbool}}{\der{\Gamma}{\neg\e}{\tbool}}{}
  \\ \\
  \myrule{\der{\Gamma}{\e_1}{\S}\quad\der{\Gamma}{\e_2}{\S}}{\der{\Gamma}{\e_1\oplus\e_2}{\S}}{}
  \qquad
  \myrule{\der{\Gamma}{\e_1}{\tint}\quad\der{\Gamma}{\e_2}{\tint}}{\der{\Gamma}{\e_1>\e_2}{\tbool}}{}
  \qquad
  \myrule{\der{\Gamma}{\e}{\S}\quad\S\subs\S'}{\der{\Gamma}{\e}{\S'}}{}
  \end{array}
  $}
  \caption{Typing rules for expressions.}
\label{tab:te}
  \end{mytable}
Table~\ref{tab:sync:typing} gives the typing rules for processes and multiparty sessions. Processes are typed as expected, the syntax of session types only allows input processes in external choices and output processes in the branches of conditionals. 
We need to assure that processes in external choices offer different labels. For this reason rule \rln{[t-in-choice]} types both inputs and external choices. With two separate rules:
\begin{myformula}{\myrule{\der{\Gamma,x:\S}{\PP}\T}{\der\Gamma{{\procin  \q   {\ell(\x)} \PP}}{\q?\ell(\S).\T}}{[t-in]}
   \qquad
  \myrule{\der{\Gamma}{\PP_1}{\T_1}~~\der{\Gamma}{\PP_2}{\T_2}}{\der\Gamma{{{\PP_1}+{\PP_2}}}{\T_1\wedge \T_2}}{[t-choice]}}\end{myformula}
we could derive \begin{myformula}{\der{}{\procin  \q   {\ell_1(\x)} \inact+\procin  \q   {\ell_2(\x)} \inact+\procin  \q   {\ell_1(\x)} {\procout \q {\ell_5(\true)}\inact}}{\q?\ell_2(\tint).\tend\wedge\q?\ell_1(\tint).\q!\ell_5(\tbool).\tend}.}\end{myformula}
\indent
In order to type a session,  rule  \rln{[t-sess]} requires that the processes in parallel can play as participants of a whole communication protocol or the terminated process, i.e. their types are projections of a unique global type. We define the set $\participant\G$ of participants of a global type $\G$ as follows:
\begin{myformula}{
\begin{array}{c}
 \participant{\Gvti \pp\q \ell \S {\G}}=\{\pp,\q\}\cup\participant{\G_i} ~(i\in I)\footnotemark\\
 \participant{\mu\ty. \G}=\participant{\G} 
\qquad
  \participant\ty=\emptyset \qquad \participant\tend=\emptyset \end{array}
}\end{myformula}
  The condition $\participant \G\subseteq\{\pp_i\mid i\in I\}$ allows to type also sessions containing $\pa\pp\inact$, a property needed to assure invariance of types under structural congruence.\footnotetext{The projectability of $\G$ assures $\participant{\G_i}=\participant{\G_j}$ for all $i,j\in I$.}  
  \begin{mytable}{}
\centerline{$
\begin{array}{c}
  \myrule{\forall i\in I\quad\der{\Gamma,x:\S_i}{\PP_i}\T_i}{\der\Gamma{\sum\limits_{i\in I}{\procin  \q   {\ell_i(\x)} \PP_i}}{\bigwedge_{i\in I}\q?\ell_i(\S_i).\T_i}}{[t-in-choice]}
   \qquad
   \der\Gamma{ \inact}\tend~~\rln{[t-$\inact$]}
  \\\\
  \myrule{\der\Gamma\e\S~~\  \der{\Gamma}{\PP}\T}{\der\Gamma{{\procout  \q   {\ell(\e)} \PP}}{\q!\ell(\S).\T}}{[t-out]}
   \qquad
   \myrule{
  \der\Gamma\e\tbool~~\der{\Gamma}{\PP_1}{\T_1}~~\der{\Gamma}{\PP_2}{\T_2}}{\der\Gamma{{\cond{\e}{\PP_1}{\PP_2}}}{\T_1\vee \T_2}}{[t-cond]}
  \\\\
  \myrule{\der{\Gamma,X:\T}{\PP}{\T}}
  {\der\Gamma{\mu X.\PP}\T}{[t-rec]}
  \quad
   \der{\Gamma,X:\T}{X}{\T}~~\rln{[t-var]}  
       \quad
  \myrule{\der\Gamma{\PP}\T \qquad \T\subt \T'}
   {\der\Gamma{\PP}\T'}{[t-sub]}
  \\\\
  \myrule{\forall i\in I\quad\der{}{\PP_i}{\proj\G{\pp_i}}  
    \quad \participant \G\subseteq\{\pp_i\mid i\in I\}}
  {\der{}{\prod\limits_{i\in I}\pa {\pp_i}\PP_i}\G
  }{[t-sess]}
    \end{array}
$}
\caption{\label{tab:sync:typing} Typing rules for processes and sessions.}
\end{mytable}
 
 \smallskip
 
The proposed type system for multiparty sessions enjoys type preservation under reduction (subject reduction) and the safety property that a typed multiparty session will never get stuck. The remaining of this section is devoted to the proof of these properties. 

\smallskip

As usual we start with an inversion and a substitution lemmas.

\begin{mylemma}{Inversion lemma}\label{lem:Inv_S}
 \begin{myenumerate}
\item
Let $\der{\Gamma}{\PP}{\T}$.
\begin{myenumerate}\label{lem:Inv_S1}
\item\label{lem:Inv_S1a} 
If $\PP = \sum\limits_{i\in I}\procin  {\pp_i}   {\ell_i(\x)} {\Q_i}$, then
$ \bigwedge_{i\in I}\procin {\pp_i} {\ell_i{(\S_i)}}{\T_i}\subt\T $ and $ \der{\Gamma, x:\S_i}{ \Q_i}{\T_i} $.
\item\label{lem:Inv_S1b} If $\PP = \procout  \pp   {\ell(\e)} \Q$, then
$\procout \pp {\ell{(\S)}}{\T'}\subt\T$ and $ \der{\Gamma}{\e}{\S}$ and  $\der{\Gamma}{ \Q}{\T'} $.
\item \label{lem:Inv_S1d}
If $\PP = \cond{\e}{\Q_1}{\Q_2} $, then
$\T_1\vee\T_2\subt\T$ and
  $\der{\Gamma}{\Q_1}{\T_1}$ and $\der{\Gamma}{\Q_2}{\T_2}$.
\item  If $\PP = \mu X.\Q $, then $ \der{\Gamma, X:\T}{\Q}{\T} $.
\item  
If $\PP = X$, then $\Gamma = \Gamma', X:\T'$ and $\T'\subt\T$.
\item  
If $\PP = \inact$, then $\T = \tend$.
\end{myenumerate}
\item\label{lem:Inv_S2} If $\der{} {\prod\limits_{i\in I}\pa{\pp_i}{\PP_i}}\G$, then 
$\der{}{\PP_i}{\proj\G{\pp_i}}$ for all $i\in I$ and $\participant\G\subseteq\set{\pp_i\mid i \in I}$.
\end{myenumerate}
\end{mylemma}
\begin{proof}
By induction on type derivations.
\end{proof}

\begin{mylemma}{Substitution lemma}\label{lem:Subst_S}
If 
$\der{\Gamma, x:\S}{\PP}{\T} $ and $\der{\Gamma}{\val}{\S}$, then 
$\der{\Gamma}{\PP \sub{\val}{\x}}{\T}$.
\end{mylemma}
\begin{proof}
By structural induction on $\PP$.
\end{proof}

In order to state subject reduction we need to formalise how global types are modified by reducing multiparty sessions. 

\begin{mydefinition} 
\begin{enumerate}
\item The {\em consumption} of the communication $\pp\lts{\ell}\q$ for the global type $\G$ (notation $\redG\G\pp\ell\q$) is the global type inductively defined by: 
\begin{myformula}{\begin{array}{c}\redG{(\Gvti  \pr\ps \ell \S {\G})}\pp\ell\q=\begin{cases}
  \G_{i_0}    & \text{if }\pr=\pp, \ps=\q, \ell_{i_0}=\ell\\
   \Gvtir  \pr \ps \ell \S{\G}   & \text{otherwise}
\end{cases}\\\\
\redG{(\mu\ty.\G)}\pp\ell\q= \mu\ty.\redG{\G}\pp\ell\q
\end{array}}\end{myformula}
\item
The reduction of global types is the smallest pre-order relation closed under the rule:
\begin{myformulaC}{\G\RedG\redG\G\pp{\ell}\q}\end{myformulaC}
\end{enumerate}
\end{mydefinition}
\vspace{-10pt}
Notice that $\redG\tend\pp\ell\q$ and $\redG\ty\pp\ell\q$ are undefined. 
It is easy to verify that, if $\G$ is projectable and $\redG\G\pp\ell\q$ is defined, then the global type $\redG\G\pp\ell\q$ is projectable. The following lemma shows other properties of consumption that are essential in the proof of subject reduction.
\begin{lemma}\label{lem:erase}
 If $\tout\q\ell{\S}.{\T}\leq\proj\G\pp$ and $\tin\pp\ell{\S}.{\T'}\wedge\T''\leq\proj\G\q$, then $\T\leq\proj{(\redG\G\pp\ell\q)}\pp$ and\\ $\T'\leq\proj{(\redG\G\pp\ell\q)}\q$. Moreover $\proj\G\pr=\proj{(\redG\G\pp\ell\q)}\pr$ for $\pr\not=\pp$, $\pr\not=\q$.
\end{lemma}
\begin{proof} By induction on $\G$ and by cases on the definition of $\redG\G\pp\ell\q$.  Notice that    $\G$ can only be $\Gvti {\ps_1} {\ps_2} \ell \S {\G}$ with either $\ps_1=\pp$ and $\ps_2=\pq$ or $\set{\ps_1,\ps_2}\cap\set{\pp,\q}=\emptyset$, since otherwise the types in the statement of the lemma could not be subtypes of the given projections of $\G$. 

 If $\G=\Gvti  \pp \q \ell \S {\G}$, then $\proj\G\pp=\bigvee_{i\in I}\tout \q{\ell_i}{\SortI{\S}{i}}.\proj{\G_i}\pp$  and 
            $\proj\G\q=\bigwedge_{i\in I}\tin\pp{\ell_i}{\SortI{\S}{i}}.\proj{\G_i}\q$. From $\tout\q\ell{\S}.{\T}\leq\bigvee_{i\in I}\tout \q{\ell_i}{\SortI{\S}{i}}.\proj{\G_i}\pp$ we get $\ell=\ell_{i_0}$ and $\T\leq\proj{\G_{i_0}}\pp$ for some $i_0\in I$. From $\tin\pp\ell{\S}.{\T'}\wedge\T''\leq\bigwedge_{i\in I}\tin\pp{\ell_i}{\SortI{\S}{i}}.\proj{\G_i}\q$ and $\ell=\ell_{i_0}$ we get $\T'\leq\proj{\G_{i_0}}\q$. We get $\T\leq\proj{(\redG\G\pp\ell\q)}\pp$ and $\T'\leq\proj{(\redG\G\pp\ell\q)}\q$, since $\proj{(\redG\G\pp\ell\q)}\pp=\proj{\G_{i_0}}\pp$ and $\proj{(\redG\G\pp\ell\q)}\q=\proj{\G_{i_0}}\q$. If $\pr\not=\pp$, $\pr\not=\q$, then by definition of projection $\proj\G\pr=\proj{\G_{i_0}}\pr$ for an arbitrary $i_0\in I$, and then $\proj\G\pr=\proj{(\redG\G\pp\ell\q)}\pr$ by definition of consumption.

 If  $\G=\Gvti  {\ps_1}{\ps_2} \ell \S {\G}$ and $\set{\ps_1,\ps_2}\cap\set{\pp,\q}=\emptyset$, then $\proj\G\pp=\proj{\G_{i_0}}\pp$  and 
            $\proj\G\q=\proj{\G_{i_0}}\q$ for an arbitrary $i_0\in I$. By  definition of consumption \begin{myformula}{\redG\G\pp\ell\q=\Gvtir  {\ps_1}{\ps_2 }\ell \S {\G},}\end{myformula}which implies $\proj{(\redG\G\pp\ell\q)}\pp=\proj{(\redG{\G_{i_0}}\pp\ell\q)}\pp$ and $\proj{(\redG\G\pp\ell\q)}\q=\proj{(\redG{\G_{i_0}}\pp\ell\q)}\q$. Notice that the choice of $i_0$ does not modify the projection, by definition of projectability.  We get $\tout\q\ell{\S}.{\T}\leq\proj{\G_{i_0}}\pp$ and $\tin\pp\ell{\S}.{\T'}\wedge\T''\leq\proj{\G_{i_0}}\q$, which imply by induction $\T\leq\proj{(\redG{\G_{i_0}}\pp\ell\q)}\pp$ and $\T'\leq\proj{(\redG{\G_{i_0}}\pp\ell\q)}\q$.\\ If $\pr=\ps_1$, then $\proj\G\pr=\bigvee_{i\in I}\tout {\ps_2}{\ell_i}{\SortI{\S}{i}}.\proj{\G_i}\pr$ and\\ \centerline{$\proj{(\redG\G\pp\ell\q)}\pr=\bigvee_{i\in I}\tout {\ps_2}{\ell_i}{\SortI{\S}{i}}.\proj{(\redG{\G_i}\pp\ell\q)}\pr$,} so we conclude since by induction $\proj{\G_i}\pr=\proj{(\redG{\G_i}\pp\ell\q)}\pr$ for all $i\in I$.\\
            If $\pr=\ps_2$, then $\proj\G\pr=\bigwedge_{i\in I}\tin {\ps_1}{\ell_i}{\SortI{\S}{i}}.\proj{\G_i}\pr$ and\\ \centerline{$\proj{(\redG\G\pp\ell\q)}\pr=\bigwedge_{i\in I}\tin {\ps_1}{\ell_i}{\SortI{\S}{i}}.\proj{(\redG{\G_i}\pp\ell\q)}\pr$,} so we conclude since by induction $\proj{\G_i}\pr=\proj{(\redG{\G_i}\pp\ell\q)}\pr$ for all $i\in I$.\\
            If $\pr\not\in\set{\ps_1,\ps_2}$, then $\proj\G\pr=\proj{\G_{i_0}}\pr$ and $\proj{(\redG\G\pp\ell\q)}\pr=\proj{(\redG{\G_{i_0}}\pp\ell\q)}\pr$ for an arbitrary $i_0\in I$. We can conclude using induction.  
\end{proof}

We can now prove subject reduction.
  
\begin{mytheorem}{Subject reduction}\label{thm:SR}
  If $\der{}\N\G$ and $\N \red^* \N'$, then $\der{}{\N'}{\G'}$ for some $\G'$ such that $\G\RedG\G'$.
\end{mytheorem}

\begin{proof}
 By induction on the multiparty session reduction. We only consider the case of rule 
$\rulename{r-comm}$ as premise of rule $\rulename{r-context}$.  In this case \myformula{\N\equiv \pa\pp{\sum\limits_{i\in I} \procin{\q}{\ell_i(\x)}{\PP_i}}\; \pc \;\pa\q{\procout \pp {\ell_j(\e)} \PP}\pc\prod\limits_{l\in L} \pa{\pp_l}{\Q_l}} and  
\myformulaA{\N'\equiv
    \pa\pp{\PP_j}\sub{\val}{\x}\;\pc\;\pa\q\PP\pc\prod\limits_{l\in L} \pa{\pp_l}{\Q_l},} where $j\in I$, $\eval{\e}{\val}$. By Lemma~\ref{lem:Inv_S}(\ref{lem:Inv_S2}) $\der{}\N\G$ implies $\der{}{\sum\limits_{i\in I} \procin{\q}{\ell_i(\x)}{\PP_i}}{\proj\G\pp}$, and $\der{}{\procout \pp {\ell_j(\e)} \PP}{\proj\G\q}$, and $\der{}{\Q_l}{\proj\G{\pp_l}}$ for $l\in L$. By Lemma~\ref{lem:Inv_S}($\ref{lem:Inv_S1a}$) $\bigwedge_{i\in I}\tin \q {\ell_i}{\S_i}.{\T_i}\subt\proj\G\pp$ and $ \der{ x:\S_i}{ \PP_i}{\T_i} $ for $i\in I$. By Lemma~\ref{lem:Inv_S}($\ref{lem:Inv_S1b}$) $\tout \pp {\ell_j}{\S}.{\T}\subt\proj\G\q$ and $ \der{}{ \e}{\S} $ and $ \der{}{ \PP}{\T} $. 
 From $\bigwedge_{i\in I}\tin \q {\ell_i}{\S_i}.{\T_i}\subt\proj\G\pp$ and $\tout \pp {\ell_j}{\S}.{\T}\subt\proj\G\q$ we get $\S_j=\S$. By Lemma~\ref{lem:Subst_S} $ \der{ x:\S}{ \PP_j}{\T_j} $ and $ \der{}{ \e}{\S} $ and $\eval{\e}{\val}$ imply 
 $ \der{}{ \PP_j\sub\val\x}{\T_j} $. Then we choose $\G'=\redG\G\pp{{\ell_j}}\q$, since Lemma~\ref{lem:erase} gives $\T_j\leq\proj{(\redG\G\pp{{\ell_j}}\q)}\pp$ and $\T\leq\proj{(\redG\G\pp{{\ell_j}}\q)}\q$ and the same projections for all other participants of $\G$.
\end{proof}

To show progress a lemma on canonical forms is handy. The proof easily follows from the inspection of the typing rules. 

\begin{mylemmaB}{Canonical forms}\label{cf}
\begin{myenumerateB}
\item If $\der{}\PP{\bigwedge_{i\in I}\tin \pp {\ell_i}{\S_i}.{\T_i}}$, then $\PP=\sum\limits_{i\in I'} \procin{\pp}{\ell_i(\x)}{\PP_i}$ with $I\subseteq I'$.
\item If $\der{}\PP{\bigvee_{i\in I}\tout \pp {\ell_i}{\S_i}.{\T_i}}$, then $\pa\pq\PP\red^*\pa\pq{\procout \pp {\ell_j(\e)} \Q}$ with $j\in I$. 
\end{myenumerateB}
\end{mylemmaB}

\begin{mytheoremA}{Progress}\label{thm:progress}
  If $\der{}\N\G$, then  either $\N\equiv\pa\pp\inact$ or $\N\red\N'$. 
\end{mytheoremA}
\begin{proof} If $\G=\tend$, then $\N\equiv\pa\pp\inact$ by Lemma~\ref{lem:Inv_S}(\ref{lem:Inv_S2}). If $\G=\Gvti  \pp\pq \ell \S {\G}$, then\\ \centerline{$\N\equiv\pa\pp\PP\pc\pa\q\Q\pc\N''$} and $\der{}\PP{\bigvee_{i\in I}\tout \pq {\ell_i}{\S_i}.{\proj{\G_i}\pp}}$ and $\der{}\Q{\bigwedge_{i\in I}\tin \pp {\ell_i}{\S_i}.{\proj{\G_i}\q}}$ again by Lemma~\ref{lem:Inv_S}(\ref{lem:Inv_S2}). By Lemma~\ref{cf} $\PP=\sum\limits_{i\in I'} \procin{\pp}{\ell_i(\x)}{\PP_i}$ with $I\subseteq I'$ and $\pa\pq\Q\red^*\pa\pq{\procout \pp {\ell_j(\e)}{ \Q'}}$ with $j\in I$. Therefore,  if $\eval\e\val$, then  $\N\red^*\pa\pp\PP\pc\pa\q{\procout \pp {\ell_j(\e)}{ \Q'}}\pc\N''\red\pa\pp{\PP_j\sub\val\x}\pc\pa\q{\Q'}\pc\N''$. 
\end{proof}

The safety property that a typed multiparty session will never get stuck is a consequence of subject reduction and progress.

\begin{mytheorem}{Safety}\label{thm:Safety}
  If $\der{}\N\G$, then it does not hold $\stuck{\M}.$ 
  \end{mytheorem}

\mysection{Operational Preciseness}\label{sec:op}

We adapt the notion of operational preciseness~\cite{BHLN12,cdy14,DG14} to our calculus. 
\begin{mydefinitionA}\label{def:preciseness}
  A subtyping relation is {\em operationally precise} if  for any two types $\T$ and $\T'$ the following equivalence holds:
  \begin{center}
$\T \subt \T'$ if and only if there are no $\PP,\pp,\M$ such that:\\[3pt]
$\bullet~\der{}{\PP}{\T}$; and
$\qquad\bullet~\der{} {\Q}{\T'}$ implies $\vdash\pa\pp \Q\pc\M $; and
$\qquad\bullet~\stuck{\pa\pp{\PP\pc\M}}$.
    \end{center}
\end{mydefinitionA}  
  
  The {\em operational soundness}, i.e. if  for all $\Q$ such that $\der{}{\Q}{\T'}$ implies $\vdash\pa\pp \Q\pc\M$, then $\pa\pp{\PP\pc\M}$ is not stuck, follows from the subsumption rule \rln{[t-sub]} and the safety theorem, Theorem~\ref{thm:Safety}. 

To show the vice versa, it is handy to define the set $\participant\T$ of participants of a session type $\T$ as follows
\begin{myformula}{
\begin{array}{c}
 \participant{\bigwedge_{i\in I}\tin\pp{\ell_i}{\S_i}. \T_i}= \participant{\bigvee_{i\in I}\tout\pp{\ell_i}{\S_i}.\T_i}=\{\pp\}\cup\bigcup_{i\in I}\participant{\T_i}\\
 \participant{\mu\ty. \T}=\participant{\T}\qquad
   \participant\ty=\participant\tend=\emptyset 
\end{array}
}\end{myformula}
The proof of {\em operational completeness} comes in four  steps. 

\begin{myitemize}
\item {\bf [Step~1]} \ We characterise the negation of the subtyping relation 
by inductive rules (notation $\nsubt$). 

\item {\bf [Step~2]} \ For each type $\T$ and participant $\pp\not\in\participant\T$, we 
define a {\em characteristic global type} $\CG\T\pp$ such that $\proj{\CG\T\pp}\pp=\T$. 

\item {\bf [Step~3]} \ For each type $\T$, we 
define a {\em characteristic process} $\CP\T$ typed by $\T$,  
which offers the series of interactions described by $\T$.

\item {\bf [Step~4]} \ We prove that if $\T\nsubt \T'$, then $\stuck{\pa\pp\CP{\T}\pc\prod\limits_{1\leq i\leq n} \pa{\pp_i}{\CP{\T_i}}}$, where $\participant{\T'}=\set{\pp_1,\ldots,\pp_n}$,  and $\T_i=\proj {\CG{\T'}\pp} {\pp_i}$ for $1\leq i\leq n$. Hence we achieve completeness by choosing $\PP=\CP\T$ and 
$\N=\prod\limits_{1\leq i\leq n} \pa{\pp_i}{\CP{\T_i}}$ 
in the definition of preciseness (Definition~\ref{def:preciseness}). 
\end{myitemize}

\noindent \myparagraph{Negation of subtyping} 
\begin{mytable}{}
\centerline{$
\begin{array}{c}
  \inferrule[\rulename{nsub-endL}]
  {\T \neq \tend} {\T \nsubt \tend}{}
\qquad 
\inferrule[\rulename{nsub-endR}]
  {\T \neq \tend} {\tend \nsubt \T}{}
\qquad
\inferrule[\rulename{nsub-diff-part}]
{\pp \neq \pq \qquad \dag,\ddag \in \{?,!\}}
{\tdag\pp{\ell_1}{\S_1}.  \T_1 \nsubt  \tddag\pq{\ell_2}{\S_2}. \T_2} 
\qquad
\\\\
\inferrule[\rulename{nsub-out-in}]
{}
{\tout\pp{\ell_1}{\S_1}.  \T_1 \nsubt  \tin\pp{\ell_2}{\S_2}. \T_2} 
\qquad
\inferrule[\rulename{nsub-in-out}]
{}
{\tin \pp {\ell_1}{\S_1}.  \T_1 \nsubt  \tout\pp{\ell_2}{\S_2}. \T_2} 
\\\\
\inferrule[\rulename{nsub-in-in}]
{\ell_1 \neq \ell_2 \;\;  \text{or} \; \S_2 \not\subs \S_1 \;\;\text{or} \;\;\T_1\nsubt \T_2}
{\tin\pp{\ell_1}{\S_1}.  \T_1 \nsubt  \tin\pp{\ell_2}{\S_2}. \T_2} 
\qquad
\inferrule[\rulename{nsub-out-out}]
{\ell_1 \neq \ell_2 \;\;\makebox{or} \;\;\S_1 \not\subs \S_2\;\;\text{or}\;\; \T_1\nsubt \T_2}
{\tout\pp{\ell_1}{\S_1}.  \T_1 \nsubt  \tout\pp{\ell_2}{\S_2}. \T_2} 
\\\\
\inferrule[\rulename{nsub-intR}]
{ \T \nsubt \T_1 \text{ or } \T\nsubt \T_2}
{\T \nsubt \T_1 \wedge \T_2}
\quad
\inferrule[\rulename{nsub-uniL}]
{\T_1 \nsubt \T \text{ or }\T_2\nsubt \T}
{\T_1 \vee \T_2 \nsubt \T}
\quad
\inferrule[\rulename{nsub-intL-uniR}]
{\forall i\in I ~\forall j\in J ~\T_i \nsubt \T'_j}
{\bigwedge_{ i\in I}\T_i\nsubt \bigvee_{ j\in J}\T'_j}
\end{array}
$}
\caption{\label{tab:subtyping_negation}Negation of subtyping}
\end{mytable}
Table~\ref{tab:subtyping_negation} gives the 
negation of subtyping, which uses the negation of subsorting $\not\subs$ defined as expected. These rules say that a type different from $\tend$ cannot be compared to $\tend$, two input or output types with different participants, or different labels, or with sorts or continuations which do not match, cannot be compared. The rules in the last line just 
take into account the set theoretic properties of intersection and union. One can show that either $\T\subt \T'$ or $\T\nsubt\T'$ holds for two arbitrary types $\T,\T'$.

\begin{mylemmaA}
  $\T\nsubt\T'$ is the negation of $\T\subt \T'$. 
\end{mylemmaA}
\begin{proof} 
If $\T \nsubt \T'$, then we can show
$\T \not\subt \T'$ by induction on the
derivation of $\T \nsubt \T'$. %
We develop just two cases (the others are similar):
\begin{myitemize}
\item%
  base case 
  \rulename{nsub-diff-part}.\quad 
  Then, %
  $\T = \tdag\pp{\ell_1}{\S_1}.  \T_1$ %
  and %
  $\T' = \tddag\pq{\ell_2}{\S_2}. \T_2$ with $\pp \neq \pq$ and $\dag,\ddag \in \{?,!\}$. %
  We can verify that $\T$ and $\T'$ do not
  match %
  the conclusion of \rulename{sub-end}, nor \rulename{sub-in}, nor
  \rulename{sub-out} %
  --- hence, we conclude $\T \not\subt \T'$;
\item%
  inductive case \rulename{nsub-intL-uniR}.\quad
  Then, %
  $\T = \bigwedge_{ i\in I}\T_i$ %
  and %
  $\T' = \bigvee_{ j\in J}\T_j'$; %
  moreover, %
  \mbox{$\forall i\in I ~\forall j\in J:$} $\T_i \nsubt \T_j'$ %
  --- and thus, by the induction hypothesis, %
  $\T_i \not\subt \T'_j$. %
  We now notice that $\T \subt \T'$ %
  could only possibly hold by rule \rulename{sub-in} when $J$ is a singleton and 
  by rule \rulename{sub-out} when $I$ is a singleton%
  --- but, since $\T_i \not\subt \T'_j$, %
  at least one of the coinductive premises of such rules is 
  not satisfied. %
  Hence, we conclude $\T \not\subt \T'$.
\end{myitemize}
Vice versa, assume $\T \not\subt \T'$: %
if we try to apply the subtyping rules to show
$\T \subt \T'$, %
we will ``fail'' after $n$ derivation steps, %
by finding two types $\T_1, \T_2$ %
whose syntactic shapes do \emph{not} match the conclusion of
\rulename{sub-end}, nor \rulename{sub-in}, nor \rulename{sub-out}. %
We prove $\T \nsubt \T'$ by induction on $n$:
\begin{myitemize}
\item%
  base case $n=0$.\quad %
  The derivation ``fails'' immediately, %
  i.e.~$\T_1=\T$ and
  $\T_2=\T'$. %
  By cases on the possible shapes of $\T$ and $\T'$, %
  we obtain %
  $\T \nsubt \T'$ %
  by one of the rules %
  \rulename{nsub-endL}, \rulename{nsub-endR}, \rulename{nsub-diff-part},
  \rulename{nsub-out-in}, \rulename{nsub-in-out},
  \rulename{nsub-in-in}, \rulename{nsub-out-out};
\item%
  inductive case $n=m+1$.\quad %
  The shapes of $\T,\T'$ match the conclusion of
 \rulename{sub-in} (resp.~\rulename{sub-out}), %
  but there is some coinductive premise $\T_1 \subt \T_2$ %
  whose sub-derivation ``fails'' after $m$ steps. %
  By the induction hypothesis, %
  we have $\T_1 \nsubt \T_2$: %
  therefore, we can derive $\T \nsubt \T'$
  by one of the rules \rulename{nsub-in-in} or \rulename{nsub-intR} %
  (or~\rulename{nsub-out-out} or \rulename{nsub-uniL}) or \rulename{nsub-intL-uniR}.
\end{myitemize}
\end{proof}

\noindent \myparagraph{Characteristic global types} 
\begin{mytable}{}
\centerline{$
\begin{array}{c}
 \CGZ{\bigwedge_{i\in I}\tin{\pp_{j_0}}{\ell_i}{\S_i}. \T_i}\pp{\set{\pp_j}_{1\leq j\leq n}}=\Gvti {\pp_{j_0}}\pp \ell \S {\G^{j_0}}\\
\CGZ{\bigvee_{i\in I}\tout{\pp_{j_0}}{\ell_i}{\S_i}.\T_i}\pp{\set{\pp_j}_{1\leq j\leq n}}=\Gvti \pp{\pp_{j_0}} \ell \S {\G^{j_0}}\\
 \CGZ{\mu\ty. \T}\pp{\set{\pp_j}_{1\leq j\leq n}}=\mu\ty. \CGZ{\T}\pp{\set{\pp_j}_{1\leq j\leq n}}\\\CGZ\ty\pp{\set{\pp_j}_{1\leq j\leq n}}=\ty \qquad\qquad \CGZ\tend\pp{\set{\pp_j}_{1\leq j\leq n}}=\tend 
\end{array}$\\}
\centerline{$
\begin{array}{lll}
\G_i^{j_0}&=&\pp_{j_0}\to\pp_{j_0+1}:\ell_i(\tbool).\ldots\pp_{n-1}\to\pp_n:\ell_i(\tbool).\pp_n\to\pp_1:\ell_i(\tbool).\\&&\pp_1\to\pp_2:\ell_i(\tbool).\ldots.\pp_{j_0-1}\to\pp_{j_0}:\ell_i(\tbool).\CGZ{\T_i}\pp{\set{\pp_j}_{1\leq j\leq n}}\end{array}
$}\caption{The function $\CGZ{\T}\pp{\set{\pp_j}_{1\leq j\leq n}}$.}\label{cgtd} \end{mytable}
The characteristic global type $\CG{\T}\pp$ of the type $\T$ for the participant $\pp$ describes  the communications between $\pp$ and all participants in $\participant\T$ following $\T$. In fact after each communication involving $\pp$ and some $\q\in\participant\T$, $\q$ starts a cyclic communication involving all participants in $\participant\T$ both as receivers and senders. 
This is needed for getting both a projectable global type and a stuck session, see the proof of Theorem~\ref{prec} and Examples~\ref{e1} and~\ref{e2}. More precisely, 
we define the characteristic global type $\CG{\T}\pp$ of the type $\T$ for the participant $\pp\not\in\participant\T$ as $\CG{\T}\pp=\CGZ{\T}\pp{\participant\T}$, where $\CGZ{\T}\pp{\set{\pp_j}_{1\leq j\leq n}}$ is given in Table~\ref{cgtd}. 
\begin{myexampleC}\label{e1}
Some characteristic global types are projectable thanks to the cyclic communication. Take for example $\T=\tout\q{\ell_1}\tnat.\tin\pr{\ell_2}\tint.\tend\vee\tout\q{\ell_3}\tint.\tend.$ Without the cyclic communication we would get the global type $\G=\pp\to\q:\set{\ell_1(\tnat).\pr\to\pp:\ell_2(\tint).\tend,\ell_3(\tint).\tend}$ and $\proj{\G}\pr=\tout\pp{\ell_2}\tint.\tend \I\tend$ is undefined. Instead
\myformula{\begin{array}{l}
\CG{\T}\pp=\pp\to\q:\{\ell_1(\tnat). \q\to\pr: \ell_1(\tbool). \pr\to\q: \ell_1(\tbool).\\
\phantom{\CG{\T}\pp=\pp\to\q:\{}\pr\to\pp:\ell_2(\tint). \pr\to\q:\ell_2(\tbool).\q\to\pr:\ell_2(\tbool).\tend,\\
\phantom{\CG{\T}\pp=\pp\to\q:\{}\ell_3(\tint). \q\to\pr: \ell_3(\tbool). \pr\to\q: \ell_3(\tbool).\tend\}\\[3pt]
\proj{\CG{\T}\pp}\pr=\tin\q{\ell_1}\tbool.\tout\q{\ell_1}\tbool.\tout\pp{\ell_2}\tint.\tout\q{\ell_2}\tbool.\tin\q{\ell_2}\tbool.\tend\wedge \\
\phantom{\proj{\CG{\T}\pp}\pr= \;}\tin\q{\ell_3}\tbool.\tout\q{\ell_3}\tbool.\tend\end{array}}
\end{myexampleC}
It is easy to verify that $\proj{\CG{\T}\pp}\pp=\T$ and $\proj{\CG{\T}\pp}\q$ is defined for all $\q\in\participant\T$ by induction on the definition of characteristic global types. 

%

\noindent \myparagraph{Characteristic processes}
We define the characteristic process $\CP{\T}$ of the type $\T$ by using the operators ${\tt succ}$, ${\tt neg}$,  and $\neg$ to check if the received values are of the right sort and exploiting the correspondence between external choices and intersections, conditionals and unions. Conditionals also allow the evaluation of expressions which can be stuck. The definition of $\CP{\T}$ by induction on $\T$ is given in Table~\ref{cpt}.
\begin{mytable}{}
  \centerline{$
\CP{\T} = 
  \left\{
  \begin{array}{ll}
  \procin  \pp {\ell(\x)}\cond{\fsucc\x>0}{\CP{\T'}}{\CP{\T'}}            & \text{ if }\T=\pp?\ell(\tnat).\T', \\
  \procin  \pp {\ell(\x)}\cond{\fsqrt\x>0}{\CP{\T'}}{\CP{\T'}}            & \text{ if }\T=\pp?\ell(\tint).\T', \\
  \procin  \pp {\ell(\x)}\cond{\neg\x}{\CP{\T'}}{\CP{\T'}}            & \text{ if }\T=\pp?\ell(\tbool).\T', \\
    \procout \pp {\ell(5)} {\CP{\T'}}                                     & \text{ if }\T=\pp!\ell(\tnat).\T',\\
  \procout \pp {\ell(-5)} {\CP{\T'}}                                   & \text{ if }\T=\pp!\ell(\tint).\T',\\
  \procout \pp {\ell(\true)} {\CP{\T'}}                                & \text{ if }\T=\pp!\ell(\tbool).\T',\\
  \CP{\T_1}\external\CP{\T_2}                                           & \text{ if }\T=\T_1\texternal \T_2, \\
     \cond{\true\oplus\false} {\CP{\T_1}}   {\CP{\T_2}}                                                          & \text{ if }\T=\T_1\tinternal \T_2, \\
      \mu X_\ty.\CP{\T'}                                                                & \text{ if }\T=\mu\ty.\T', \\
     X_\ty                                                                & \text{ if }\T=\ty, \\
      \inact                                                                & \text{ if }\T=\tend. 
  \end{array}
  \right.
$}
\caption{Characteristic processes}\label{cpt}
\end{mytable}
 By induction on the structure of $\CP\T$ it is easy to verify that $\vdash \CP\T:\T$.

\smallskip

We have now all the necessary machinery to show operational preciseness of subtyping.

\begin{mytheorem}{Preciseness}\label{prec}
  The synchronous multiparty session subtyping is operationally precise.
\end{mytheorem}

\begin{proof} We only need to show completeness of the synchronous multiparty session subtyping.\\
Let  $\T\subt\T'$ and $\pp\not\in\participant{\T'}=\set{\pp_i}_{1\leq i\leq n}$ and $\G=\CG{\T'}\pp$ and $\T_i=\proj \G {\pp_i}$ for $1\leq i\leq n.$\\ Then $\vdash\Q:\T'$ implies \mbox{$\vdash\pa\pp\Q\pc\prod\limits_{1\leq i\leq n} \pa{\pp_i}{\CP{\T_i}}$} by rule [\rln{t-sess}]. 
We show that  \myformula{\stuck{\pa\pp\CP{\T}\pc\prod\limits_{1\leq i\leq n} \pa{\pp_i}{\CP{\T_i}}}.} The proof is by induction on the definition of $\nsubt$. We only consider some interesting cases.

\medskip

\inferrule[\rulename{nsub-diff-part}]
{\q \neq \pp_h \qquad \dag,\ddag \in \{?,!\}}
{\tdag\q{\ell}{\S}.  \T_0 \nsubt  \tddag{\pp_{h}}{\ell'}{\S'}. \T_0'} 

\smallskip

\noindent
By definition $\CP{\T}=\q\dag\ell(\e).\PP$ for suitable $\e,\PP$. If $\q\not\in\set{\pp_i}_{1\leq i\leq n}$, then \myformulaC{\stuck{\pa\pp\CP{\T}\pc\prod\limits_{1\leq i\leq n} \pa{\pp_i}{\CP{\T_i}}},} since $\CP{\T}$ will never communicate.\\
Otherwise let $\q=\pp_j$ with $1\leq j\leq n$ and $j\not=h$. By construction $\CP{\T_h}=\pp\dual\ddag\ell'(\e_h).\PP_h$, where $\dual\ddag=\begin{cases}
  ?    & \text{if }\ddag=!\\
   !   & \text{if }\ddag=?
\end{cases}$, and $\CP{\T_k}=\procin{\pp_{f(k)}}{\ell'(\x)}{\PP_k}$, where $f(k)=\begin{cases}
k-1      & \text{if }k>1 \\
 n     & \text{if } k=1
\end{cases}$ for $1\leq k\leq n$ and $k\not=h$.  
Therefore $\pa\pp\CP{\T}\pc\prod\limits_{1\leq i\leq n} \pa{\pp_i}{\CP{\T_i}}$ cannot reduce. 

\medskip

\inferrule[\rulename{nsub-in-in}]
{\ell_1 \neq \ell_2 \;\;  \text{or} \; \S_2 \not\subs \S_1 \;\;\text{or} \;\;\T_1\nsubt \T_2}
{\tin{\pp_h}{\ell_1}{\S_1}.  \T_1 \nsubt  \tin{\pp_h}{\ell_2}{\S_2}. \T_2}

\smallskip

\noindent
A paradigmatic case is $\ell_1=\ell_2=\ell$, $\S_1=\tnat$, $\S_2=\tint$,  $\T_1=\T_2=\tend$. By definition $\participant{\T'}=\set{\pp_h}$ and $\CP{\T}=\procin  {\pp_h} {\ell(\x)}\cond{\fsucc\x>0}{\inact}{\inact}$   and $\CP{\T_h}=\procout \pp {\ell(-5)} {\inact}$. Therefore $\pa\pp\CP{\T}\pc\CP{\T_h}$ reduces
to $\pa\pp\cond{\fsucc{-5}>0}{\inact}{\inact}$, which is stuck. 

\medskip

\inferrule[\rulename{nsub-intR}]
{ \T \nsubt \T'_1 \text{ or } \T\nsubt \T'_2}
{\T \nsubt \T'_1 \wedge \T'_2}

\smallskip

\noindent 
By definition $\T'_1$ and $\T'_2$ must be intersections of inputs with the same sender, let it be $\pp_h$. Let $\G_1=\CG{\T_1'}\pp$, $\G_2=\CG{\T_2'}\pp$, 
$\PP^{(1)}_h=\CP{\proj {\G_1} {\pp_h}}$, $\PP^{(2)}_h=\CP{\proj {\G_2} {\pp_h}}$. Then by construction \myformula{\PP_h=\CP{\proj {\CG{\T_1'\wedge \T'_2}\pp} {\pp_h}}= \cond{\true\oplus\false}{\PP^{(1)}_h}{\PP^{(2)}_h}.} This implies that $\pa\pp\CP{\T}\pc\prod\limits_{1\leq i\leq n} \pa{\pp_i}{\CP{\T_i}}$ reduces to both $\pa\pp\CP{\T}\pc\pa{\pp_h}{\PP^{(1)}_h}\pc\prod\limits_{1\leq i\not=h\leq n} \pa{\pp_i}{\CP{\T_i}}$ and $\pa\pp\CP{\T}\pc\pa{\pp_h}{\PP^{(2)}_h}\pc\prod\limits_{1\leq i\not=h\leq n} \pa{\pp_i}{\CP{\T_i}}$. By induction either $\pa\pp\CP{\T}\pc\pa{\pp_h}{\PP^{(1)}_h}\pc\prod\limits_{1\leq i\not=h\leq n} \pa{\pp_i}{\CP{\T_i}}$ or $\pa\pp\CP{\T}\pc\pa{\pp_h}{\PP^{(2)}_h}\pc\prod\limits_{1\leq i\not=h\leq n} \pa{\pp_i}{\CP{\T_i}}$ is stuck, and therefore also $\pa\pp\CP{\T}\pc\prod\limits_{1\leq i\leq n} \pa{\pp_i}{\CP{\T_i}}$ is stuck.

\medskip

\inferrule[\rulename{nsub-uniL}]
{\T'_1 \nsubt \T \text{ or }\T'_2\nsubt \T}
{\T'_1 \vee \T'_2 \nsubt \T}

\smallskip

\noindent 
By definition $\T'_1$ and $\T'_2$ must be unions of outputs with the same receiver, let it be $\pp_h$. By definition $\CP{\T'_1 \vee \T'_2}= \cond{\true\oplus\false} {\CP{\T'_1}}   {\CP{\T'_2}}$. Then $\pa\pp\CP{\T'_1 \vee \T'_2}\pc\prod\limits_{1\leq i\leq n} \pa{\pp_i}{\CP{\T_i}}$ reduces to both $\pa\pp\CP{\T'_1}\pc\prod\limits_{1\leq i\leq n} \pa{\pp_i}{\CP{\T_i}}$ and $\pa\pp\CP{\T'_2}\pc\prod\limits_{1\leq i\leq n} \pa{\pp_i}{\CP{\T_i}}$. By induction\\ \centerline{either $\pa\pp\CP{\T'_1}\pc\prod\limits_{1\leq i\leq n} \pa{\pp_i}{\CP{\T_i}}$ or $\pa\pp\CP{\T'_2}\pc\prod\limits_{1\leq i\leq n} \pa{\pp_i}{\CP{\T_i}}$ is stuck,} and therefore   
$\pa\pp\CP{\T'_1 \vee \T'_2}\pc\prod\limits_{1\leq i\leq n} \pa{\pp_i}{\CP{\T_i}}$ 
is stuck too.

\medskip

\inferrule[\rulename{nsub-intL-uniR}]
{\forall l\in L ~\forall j\in J ~\T'_l \nsubt \T''_j}
{\bigwedge_{ l\in L}\T'_l\nsubt \bigvee_{ j\in J}\T''_j}

\smallskip

\noindent
If $L$ and $J$ are both singleton sets it is immediate by induction.

\noindent
If  $L$ and $J$ both contain more than one index, then by definition we can assume (without loss of generality) that $\T'_l$ for $l\in L$ are input types with the same sender, let it be $\pp_h$, and $\T''_j$ for $j\in J$ are output types with the same receiver, let it be $\pp_k$. By definition  $\CP{\T}=\sum\limits_{l\in L}\procin{\pp_h}{\ell_l(\x)}{\PP_l'}$, and $\CP{\T_k}=\sum\limits_{j\in J}\procin{\pp}{\ell_j(\x)}{\PP_j''}$ and $\CP{\T_u}=\pp_{f(u)}?\ell_j(x).\PP_u$, where $f$ is as in the case of rule \rulename{nsub-diff-part}, for \mbox{$1\leq u\leq n$} and $u\not=k$.
Therefore $\pa\pp\CP{\T}\pc\prod\limits_{1\leq i\leq n} \pa{\pp_i}{\CP{\T_i}}$ cannot reduce.

\noindent
Let $L$ contains more than one index and $J$ be a singleton set.
By definition  $\CP{\T}=\sum\limits_{l\in L}{\PP_l'}$, where ${\PP_l'}=\CP{\T'_l}$ for $l\in L$. Let us assume ad absurdum that \mbox{$\pa\pp\CP{\T}\pc\prod\limits_{1\leq i\leq n} \pa{\pp_i}{\CP{\T_i}}$} is not stuck. Then there must be $l_0\in L$ such that \mbox{$\pa\pp{\PP_{l_0}'}\pc\prod\limits_{1\leq i\leq n} \pa{\pp_i}{\CP{\T_i}}$} is not stuck, contradicting the hypothesis.

\noindent
If $L$ is a singleton set and $J$ contains more than one index, then $\T''_j$ for $j\in J$ must be unions of outputs with the same receiver, let it be $\pp_h$. 
Let $\G_j=\CG{\T_j''}\pp$ and $\PP^{(j)}_h=\CP{\proj {\G_j} {\pp_h}}$. Then $\PP_h=\CP{\proj {\CG{\bigvee_{ j\in J}\T''_j}\pp}{\pp_h}}=\sum\limits_{j\in J}{\PP^{(j)}_h}$. Let us assume ad absurdum that $\pa\pp\CP{\T}\pc\prod\limits_{1\leq i\leq n} \pa{\pp_i}{\CP{\T_i}}$ is not stuck. In this case there must be $j_0\in J$ such that 
$\pa\pp{\CP{\T}}\pc\pa{\pp_h}{\PP^{(j_0)}_h}
\pc\prod\limits_{1\leq i\not=h\leq n} \pa{\pp_i}{\CP{\T_i}}$ is not stuck, contradicting the hypothesis.
\end{proof}
\begin{myexample}\label{e2}
An example showing the utility of the cyclic communication in the definition of characteristic global types is $\T=\tout{\pp_1}{\ell_1}{\tnat}.\tout{\pp_2}{\ell_2}{\tnat}.\tend$ and $\T'=\tout{\pp_2}{\ell_2}{\tnat}.\tout{\pp_1}{\ell_1}{\tnat}.\tend$. In fact without the cyclic communication the characteristic global type of $\T'$ would be 
\myformula{\G=\pp\to\pp_2:\ell_2(\tnat).\pp\to\pp_1:\ell_1(\tnat).\tend} and then $\M=\pa{\pp_1}{\CP{\proj\G{\pp_1}}}\pc\pa{\pp_2}{\CP{\proj\G{\pp_2}}}=\pa{\pp_1}{\procin{\pp}{\ell_1(\x)}{\inact}}\pc\pa{\pp_2}{\procin{\pp}{\ell_2(\x)}{\inact}}$. Being $\CP\T=\procout{\pp_1}{\ell_1(5)}{\procout{\pp_2}{\ell_2(5)}\inact}$, the session $\pa\pp{\CP\T}\pc\M$ reduces to $\pa\pp\inact$. Instead
\myformula{\begin{array}{lll}\CG{\T'}\pp&=&\pp\to\pp_2:\ell_2(\tnat).\pp_2\to\pp_1:\ell_2(\tbool).\pp_1\to\pp_2:\ell_2(\tbool).\\
&&\pp\to\pp_1:\ell_1(\tnat).\pp_1\to\pp_2:\ell_1(\tbool).\pp_2\to\pp_1:\ell_1(\tbool).\tend,\end{array}} which implies $\CP{\proj{\CG{\T'}\pp}{\pp_1}}=\procin{\pp_2}{\ell_2(\x)}{\ldots}$ and $\CP{\proj{\CG{\T'}\pp}{\pp_2}}=\procin{\pp}{\ell_2(\x)}{\ldots}.$ It is then easy to verify that $\pa\pp{\CP\T}\pc \pa{\pp_1}{\CP{\proj{\CG{\T'}\pp}{\pp_1}}}\pc\pa{\pp_2}{\CP{\proj{\CG{\T'}\pp}{\pp_2}}}$ is stuck. 
\end{myexample}


\mysection{Operational Preciseness at Work}\label{ex}
Consider  a multiparty session with four participants: client $({\tt cl}),$ adder $({\tt add})$, increment $({\tt inc}),$ and decrement $({\tt dec})$
  \[\pa{\tt cl}{\PP_{\tt cl}} \sep  \pa{\tt add}{\PP_{\tt add}} \sep  \pa{\tt inc}{\PP_{\tt inc}} \sep  \pa{\tt dec}{\PP_{\tt dec}}.
  \]
Client sends two natural numbers to adder and expects the integer result of summation. Adder receives the two numbers and   sum them  by successively increasing the first one by 1  (done by {\tt inc}) and decreasing the second one by 1 (done by {\tt dec}). If the second summand equals 0, the first summand gives the required sum. 
Processes modelling this behaviour are the following:
\myformula{
\begin{array}{lcl}
\PP_{\tt cl}     & = &  \procout{\tt add}{\ell_1(5)}\procout{\tt add}{\ell_2(4)}\procin{\tt add}{\ell_3(x)}\inact\\
\PP_{\tt add}  & = & \procin{\tt cl}{\ell_1(y_1)}\procin{\tt cl}{\ell_2(y_2)}\mu X. \kf{if}~y_2=0~\kf{then}~
                       \procout{\tt inc}{\ell_4(\true)}\procout{\tt dec}{\ell_4(true)}
                       \procout{\tt cl}{\ell_3(y_1)}{\tend}\\ 
                       && \phantom{\procin{\tt cl}{\ell_1(y_1)}\procin{\tt cl}{\ell_2(y_2)}\mu X. \kf{if}~y_2=0~}\kf{else}~
                       \procout{\tt inc}{\ell_5(y_1)}\procin{\tt inc}{\ell_6(y_1)}\procout{\tt dec}{\ell_7(y_2)}\procin{\tt dec}{\ell_8(y_2)}X \\
\PP_{\tt inc}  & = & \mu X.\procin{\tt add}{\ell_4(\tbool)}\tend \external 
\procin{\tt add}{\ell_5(y)}\procout{\tt add}{\ell_6(y+1)}X \\
\PP_{\tt dec} & = & \mu X.\procin{\tt add}{\ell_4(\tbool)}\tend \external 
\procin{\tt add}{\ell_7(y)}\procout{\tt add}{\ell_8(y-1)}X. \\
\end{array}
}
We can extend addition to integers by changing the process  $\PP_{\tt add}$ as follows:
\myformula{
\begin{array}{lcl}
\PP'_{\tt add} & = &  \procin{\tt cl}{\ell_1(y_1)}\procin{\tt cl}{\ell_2(y_2)}\mu X. \kf{if}~y_2=0~\kf{then}~
                        \procout{\tt inc}{\ell_4(\true)}\procout{\tt dec}{\ell_4(true)} 
                         \procout{\tt cl}{\ell_3(y_1)}{\tend} \\
                        &&
                        \hfill\kf{else}~\kf{if}~y_2>0~\kf{then}~
                       \procout{\tt inc}{\ell_5(y_1)}\procin{\tt inc}{\ell_6(y_1)}\procout{\tt dec}{\ell_7(y_2)}\procin{\tt dec}{\ell_8(y_2)}X \\
                       &&\hfill\kf{else}~\procout{\tt inc}{\ell_5(y_2)}\procin{\tt inc}{\ell_6(y_2)}\procout{\tt dec}{\ell_7(y_1)}\procin{\tt dec}{\ell_8(y_1)}X.
\end{array}
}
 Process $\PP'_{\tt add}$ additionally checks if the second summand is positive. If it is not, the sum is calculated by successively increasing the second summand by 1 and decreasing the first summand by 1. 
The new multiparty session follows the global protocol
  \myformula{
    \begin{array}{l}
      {\tt cl} \to {\tt add}: \ell_1(\tint).{\tt cl} \to {\tt add}: \ell_2(\tint).\mu{\bf t}.{\tt add}\to {\tt inc}:\{\\
        \ell_4(\tbool):{\tt add}\to{\tt dec}: \ell_4(\tbool).{\tt add}\to{\tt cl}:\ell_3(\tint).\tend, \\
       \ell_5(\tint). {\tt inc}\to {\tt add}:\ell_6(\tint).
       {\tt add}\to {\tt dec}:\ell_7(\tint).{\tt dec}\to {\tt add}:\ell_8(\tint).{\bf t}\}.
    \end{array}
 }
Operational soundness of the subtyping guarantees that the summation of natural numbers will be safe after this change, as for $\tnat\subt\tint$  we have
 \myformula{
     \tout{\tt add}{\ell_1}{\tnat}.\tout{\tt add}{\ell_2}{\tnat}.\tin{\tt add}{\ell_3}{\tint}.\tend \subt  \tout{\tt add}{\ell_1}{\tint}.\tout{\tt add}{\ell_2}{\tint}.\tin{\tt add}{\ell_3}{\tint}.\tend.
}
On the other hand, by operational completeness we cannot swap sending of messages with different labels, e.g.
 \myformula{
  \T = \tout{\tt add}{\ell_1}{\tint}.\tout{\tt add}{\ell_2}{\tint}.\tend \not \subt \tout{\tt add}{\ell_2}{\tint}.\tout{\tt add}{\ell_1}{\tint}.\tend = \T'.
}
We can construct processes $Q_{\tt cl} = \procout{\tt add}{\ell_1(5)}\procout{\tt add}{\ell_2(4)}\inact$ of type $\T$ and $Q'_{\tt cl} = \procout{\tt add}{\ell_2(4)}\procout{\tt add}{\ell_1(5)}\inact$ of type $\T'$  and a  multiparty session 
 \myformula{
  \M = {\tt add} \lhd \procin{\tt cl}{\ell_2(x)}\cond{\fneg{x}>0}{\procin{\tt cl}{\ell_1(x)}{\inact}}{\procin{\tt cl}{\ell_1(x)}{\inact}}
}
such that $\pa{{\tt cl}}{Q'_{\tt cl}} \sep \M$ is well typed, while $\pa{{\tt cl}}{Q_{\tt cl}} \sep \M$ is stuck, since
the multiparty session
 \myformula{
  {\tt cl} \lhd \procout{\tt add}{\ell_1(5)}\procout{\tt add}{\ell_2(4)}\inact \sep {\tt add} \lhd \procin{\tt cl}{\ell_2(x)}\cond{\fneg{x}>0}{\procin{\tt cl}{\ell_1(x)}{\inact}}{\procin{\tt cl}{\ell_1(x).\inact}}
 }
cannot reduce because of label mismatch.
\mysection{Denotational Preciseness}\label{sec:denotation}

In $\lambda$-calculus types are usually interpreted as subsets of the domains of $\lambda$-models~\cite{bcd83,H83}. {\em Denotational preciseness} of subtyping is then:
 \begin{myformula}{\T\subt \T' \ \text{if and only if}\ \dlsqb  \T \drsqb \subseteq\dlsqb  \T' \drsqb, }\end{myformula}
using $\dlsqb  \; \drsqb $ to denote type interpretation. 

In the present context let us interpret a session type $\T$ as the set of closed processes typed by $\T$, i.e.
\myformula{\dlsqb  \T \drsqb=\{\PP~ \mid ~\vdash \PP: \T\}}
We can then show that the subtyping is denotationally precise. The subsumption rule \rln{[t-sub]} gives the denotational soundness. Denotational completeness follows from the following key property of characteristic processes:
 \begin{myformula}{\vdash \CP\T :\T'\ \text{implies}\ \T\subt \T'.}\end{myformula}
If we could derive $\vdash\CP\T :\T'$ with $\T\not\subt\T'$, 
then the multiparty session
 \begin{myformula}{ \pa\pp\CP{\T}\pc\prod\limits_{1\leq i\leq n} \pa{\pp_i}{\CP{\T_i}},}\end{myformula}
where $\participant{\T'}=\set{\pp_i}_{1\leq i\leq n}$ and $\G=\CG{\T'}\pp$ and $\T_i=\proj \G {\pp_i}$ for $1\leq i\leq n$, could be typed. Theorem~\ref{prec} shows that this process is stuck, and this contradicts the soundness of the type system. We get the desired property, which implies denotational completeness, since if $\T\not\subt\T'$, then $\CP\T\in \dlsqb  \T \drsqb$, but $\CP\T\not\in \dlsqb  \T' \drsqb$. 
  \begin{mytheorem}{Denotational preciseness}\label{thm:dpreciseness}
The subtyping relations is denotationally precise. \end{mytheorem}
  
  \mysection{Conclusion}\label{conc}
  The preciseness result of this paper 
shows a rigorousness of 
the subtyping, which is 
implemented (as a default) in most of session-based programming languages 
and tools 
\cite{event,DemangeonH11,scribble10,
HNHYH13} 
for enlarging typability.

The main technical contribution is the definition of characteristic global types, see Section~\ref{sec:op}. 
  Given a session type $\T$ and a session participant $\pp$ which does not occur in $\T$,
the associated characteristic global type expresses the communications prescribed by $\T$ between $\pp$ and the participants in $\T$. After each communication involving  $\pp$,  the characteristic global type creates a cyclic communication between all participants in $\T$. Such a cyclic communication is essential to project the characteristic global type and to generate deadlock when the 
    the subtyping relation is extended.

The subtyping considered here is sound but not complete for asynchronous multiparty sessions~\cite{HYC08}, as shown in~\cite{mostrous_yoshida_honda_esop09}. We conjecture the completeness of the subtyping defined in~\cite{mostrous_yoshida_honda_esop09} for asynchronous multiparty sessions and we are working toward this proof.

\vspace{-10pt}
\paragraph{Acknowledgments.} 
We are grateful to the anonymous reviewers for their useful remarks.

\vspace{-7pt}
\bibliographystyle{eptcs}
\bibliography{paper3_references}

\begin{thebibliography}{10}
\providecommand{\bibitemdeclare}[2]{}
\providecommand{\surnamestart}{}
\providecommand{\surnameend}{}
\providecommand{\urlprefix}{Available at }
\providecommand{\url}[1]{\texttt{#1}}
\providecommand{\href}[2]{\texttt{#2}}
\providecommand{\urlalt}[2]{\href{#1}{#2}}
\providecommand{\doi}[1]{doi:\urlalt{http://dx.doi.org/#1}{#1}}
\providecommand{\bibinfo}[2]{#2}

\bibitemdeclare{article}{bcd83}
\bibitem{bcd83}
\bibinfo{author}{Henk \surnamestart Barendregt\surnameend},
  \bibinfo{author}{Mario \surnamestart Coppo\surnameend} \&
  \bibinfo{author}{Mariangiola \surnamestart Dezani-Ciancaglini\surnameend}
  (\bibinfo{year}{1983}): \emph{\bibinfo{title}{A Filter Lambda Model and the
  Completeness of Type Assignment}}.
\newblock {\sl \bibinfo{journal}{Journal of Symbolic Logic}}
  \bibinfo{volume}{48}(\bibinfo{number}{4}), pp. \bibinfo{pages}{931--940},
  \doi{10.2307/2273659}.

\bibitemdeclare{misc}{CDL}
\bibitem{CDL}
\emph{\bibinfo{title}{{W3C} {W}{S}-{C}{D}{L}}}.
\newblock \bibinfo{howpublished}{\url{http://www.w3.org/2002/ws/chor/}}.

\bibitemdeclare{inproceedings}{cdy14}
\bibitem{cdy14}
\bibinfo{author}{Tzu-Chun \surnamestart Chen\surnameend},
  \bibinfo{author}{Mariangiola \surnamestart Dezani-Ciancaglini\surnameend} \&
  \bibinfo{author}{Nobuko \surnamestart Yoshida\surnameend}
  (\bibinfo{year}{2014}): \emph{\bibinfo{title}{On the Preciseness of Subtyping
  in Session Types}}.
\newblock In: {\sl \bibinfo{booktitle}{PPDP}}, \bibinfo{publisher}{ACM Press},
  pp. \bibinfo{pages}{135--146}, \doi{10.1145/2643135.2643138}.

\bibitemdeclare{article}{CDYP15}
\bibitem{CDYP15}
\bibinfo{author}{Mario \surnamestart Coppo\surnameend},
  \bibinfo{author}{Mariangiola \surnamestart Dezani-Ciancaglini\surnameend},
  \bibinfo{author}{Nobuko \surnamestart Yoshida\surnameend} \&
  \bibinfo{author}{Luca \surnamestart Padovani\surnameend}
  (\bibinfo{year}{2015}): \emph{\bibinfo{title}{Global Progress for Dynamically
  Interleaved Multiparty Sessions}}.
\newblock {\sl \bibinfo{journal}{Mathematical Structures in Computer Science}},
  \doi{10.1017/S0960129514000188}.
\newblock \bibinfo{note}{To appear}.

\bibitemdeclare{inproceedings}{DemangeonH11}
\bibitem{DemangeonH11}
\bibinfo{author}{Romain \surnamestart Demangeon\surnameend} \&
  \bibinfo{author}{Kohei \surnamestart Honda\surnameend}
  (\bibinfo{year}{2011}): \emph{\bibinfo{title}{Full Abstraction in a Subtyped
  pi-Calculus with Linear Types}}.
\newblock In: {\sl \bibinfo{booktitle}{CONCUR}}, {\sl \bibinfo{series}{LNCS}}
  \bibinfo{volume}{6901}, \bibinfo{publisher}{Springer}, pp.
  \bibinfo{pages}{280--296}, \doi{10.1007/978-3-642-23217-6\_19}.

\bibitemdeclare{article}{SIAM}
\bibitem{SIAM}
\bibinfo{author}{Mariangiola \surnamestart Dezani-Ciancaglini\surnameend},
  \bibinfo{author}{Ugo \surnamestart de'Liguoro\surnameend} \&
  \bibinfo{author}{Adolfo \surnamestart Piperno\surnameend}
  (\bibinfo{year}{1998}): \emph{\bibinfo{title}{A Filter Model for Concurrent
  lambda-Calculus}}.
\newblock {\sl \bibinfo{journal}{SIAM Journal on Computing}}
  \bibinfo{volume}{27}(\bibinfo{number}{5}), pp. \bibinfo{pages}{1376--1419},
  \doi{10.1137/S0097539794275860}.

\bibitemdeclare{inproceedings}{DG14}
\bibitem{DG14}
\bibinfo{author}{Mariangiola \surnamestart Dezani-Ciancaglini\surnameend} \&
  \bibinfo{author}{Silvia \surnamestart Ghilezan\surnameend}
  (\bibinfo{year}{2014}): \emph{\bibinfo{title}{Preciseness of Subtyping on
  Intersection and Union Types}}.
\newblock In: {\sl \bibinfo{booktitle}{RTATLCA}}, {\sl \bibinfo{series}{LNCS}}
  \bibinfo{volume}{8560}, \bibinfo{publisher}{Springer}, pp.
  \bibinfo{pages}{194--207}, \doi{10.1007/978-3-319-08918-8\_14}.

\bibitemdeclare{article}{GH05}
\bibitem{GH05}
\bibinfo{author}{Simon \surnamestart Gay\surnameend} \&
  \bibinfo{author}{Malcolm \surnamestart Hole\surnameend}
  (\bibinfo{year}{2005}): \emph{\bibinfo{title}{Subtyping for Session Types in
  the Pi Calculus}}.
\newblock {\sl \bibinfo{journal}{Acta Informatica}}
  \bibinfo{volume}{42}(\bibinfo{number}{2/3}), pp. \bibinfo{pages}{191--225},
  \doi{10.1007/s00236-005-0177-z}.

\bibitemdeclare{book}{harp13}
\bibitem{harp13}
\bibinfo{author}{Robert \surnamestart Harper\surnameend}
  (\bibinfo{year}{2013}): \emph{\bibinfo{title}{Practical Foundations for
  Programming Languages}}.
\newblock \bibinfo{publisher}{Cambridge University Press}.

\bibitemdeclare{inproceedings}{HNHYH13}
\bibitem{HNHYH13}
\bibinfo{author}{A.~S. \surnamestart Henriksen\surnameend},
  \bibinfo{author}{L.~\surnamestart Nielsen\surnameend},
  \bibinfo{author}{T.~\surnamestart Hildebrandt\surnameend},
  \bibinfo{author}{N.~\surnamestart Yoshida\surnameend} \&
  \bibinfo{author}{F.~\surnamestart Henglein\surnameend}
  (\bibinfo{year}{2012}): \emph{\bibinfo{title}{Trustworthy Pervasive
  Healthcare Services via Multi-party Session Types}}.
\newblock In: {\sl \bibinfo{booktitle}{FHIES}}, {\sl \bibinfo{series}{LNCS}}
  \bibinfo{volume}{7789}, \bibinfo{publisher}{Springer}, pp.
  \bibinfo{pages}{124--141}, \doi{10.1007/978-3-642-39088-3\_8}.

\bibitemdeclare{article}{H83}
\bibitem{H83}
\bibinfo{author}{J.~Roger \surnamestart Hindley\surnameend}
  (\bibinfo{year}{1983}): \emph{\bibinfo{title}{The Completeness Theorem for
  Typing Lambda-Terms}}.
\newblock {\sl \bibinfo{journal}{Theoretical Computer Science}}
  \bibinfo{volume}{22}, pp. \bibinfo{pages}{1--17},
  \doi{10.1016/0304-3975(83)90136-6}.

\bibitemdeclare{inproceedings}{scribble10}
\bibitem{scribble10}
\bibinfo{author}{Kohei \surnamestart Honda\surnameend}, \bibinfo{author}{Aybek
  \surnamestart Mukhamedov\surnameend}, \bibinfo{author}{Gary \surnamestart
  Brown\surnameend}, \bibinfo{author}{Tzu-Chun \surnamestart Chen\surnameend}
  \& \bibinfo{author}{Nobuko \surnamestart Yoshida\surnameend}
  (\bibinfo{year}{2011}): \emph{\bibinfo{title}{Scribbling Interactions with a
  Formal Foundation}}.
\newblock In: {\sl \bibinfo{booktitle}{ICDCIT}}, {\sl \bibinfo{series}{LNCS}}
  \bibinfo{volume}{6536}, \bibinfo{publisher}{Springer}, pp.
  \bibinfo{pages}{55--75}, \doi{10.1007/978-3-642-19056-8\_4}.

\bibitemdeclare{inproceedings}{HYC08}
\bibitem{HYC08}
\bibinfo{author}{Kohei \surnamestart Honda\surnameend}, \bibinfo{author}{Nobuko
  \surnamestart Yoshida\surnameend} \& \bibinfo{author}{Marco \surnamestart
  Carbone\surnameend} (\bibinfo{year}{2008}): \emph{\bibinfo{title}{Multiparty
  Asynchronous Session Types}}.
\newblock In: {\sl \bibinfo{booktitle}{POPL}}, \bibinfo{publisher}{ACM Press},
  pp. \bibinfo{pages}{273--284}, \doi{10.1145/1328438.1328472}.

\bibitemdeclare{inproceedings}{event}
\bibitem{event}
\bibinfo{author}{Raymond \surnamestart Hu\surnameend},
  \bibinfo{author}{Dimitrios \surnamestart Kouzapas\surnameend},
  \bibinfo{author}{Olivier \surnamestart Pernet\surnameend},
  \bibinfo{author}{Nobuko \surnamestart Yoshida\surnameend} \&
  \bibinfo{author}{Kohei \surnamestart Honda\surnameend}
  (\bibinfo{year}{2010}): \emph{\bibinfo{title}{Type-Safe Eventful Sessions in
  {J}ava}}.
\newblock In: {\sl \bibinfo{booktitle}{ECOOP}}, {\sl \bibinfo{series}{LNCS}}
  \bibinfo{volume}{6183}, \bibinfo{publisher}{Springer}, pp.
  \bibinfo{pages}{329--353}, \doi{10.1007/978-3-642-14107-2\_16}.

\bibitemdeclare{inproceedings}{KY13}
\bibitem{KY13}
\bibinfo{author}{Dimitrios \surnamestart Kouzapas\surnameend} \&
  \bibinfo{author}{Nobuko \surnamestart Yoshida\surnameend}
  (\bibinfo{year}{2013}): \emph{\bibinfo{title}{Globally Governed Session
  Semantics}}.
\newblock In: {\sl \bibinfo{booktitle}{CONCUR}}, {\sl \bibinfo{series}{LNCS}}
  \bibinfo{volume}{8052}, \bibinfo{publisher}{Springer}, pp.
  \bibinfo{pages}{395--409}, \doi{10.1145/1328438.1328472}.

\bibitemdeclare{techreport}{BHLN12}
\bibitem{BHLN12}
\bibinfo{author}{Jay \surnamestart Ligatti\surnameend}, \bibinfo{author}{Jeremy
  \surnamestart Blackburn\surnameend} \& \bibinfo{author}{Michael \surnamestart
  Nachtigal\surnameend} (\bibinfo{year}{2014}): \emph{\bibinfo{title}{On
  Subtyping-Relation Completeness, with an Application to Iso-Recursive
  Types}}.
\newblock \bibinfo{type}{Technical Report}, \bibinfo{institution}{University of
  South Florida}.

\bibitemdeclare{inproceedings}{mostrous_yoshida_honda_esop09}
\bibitem{mostrous_yoshida_honda_esop09}
\bibinfo{author}{Dimitris \surnamestart Mostrous\surnameend},
  \bibinfo{author}{Nobuko \surnamestart Yoshida\surnameend} \&
  \bibinfo{author}{Kohei \surnamestart Honda\surnameend}
  (\bibinfo{year}{2009}): \emph{\bibinfo{title}{Global Principal Typing in
  Partially Commutative Asynchronous Sessions}}.
\newblock In: {\sl \bibinfo{booktitle}{ESOP}}, {\sl \bibinfo{series}{LNCS}}
  \bibinfo{volume}{5502}, \bibinfo{publisher}{Springer}, pp.
  \bibinfo{pages}{316--332}, \doi{10.1007/978-3-642-00590-9\_23}.

\bibitemdeclare{inproceedings}{P11}
\bibitem{P11}
\bibinfo{author}{Luca \surnamestart Padovani\surnameend}
  (\bibinfo{year}{2011}): \emph{\bibinfo{title}{Session Types = Intersection
  Types + Union Types}}.
\newblock In: {\sl \bibinfo{booktitle}{ITRS}}, {\sl
  \bibinfo{series}{EPTCS}}~\bibinfo{volume}{45}, \bibinfo{publisher}{Open
  Publishing Association}, pp. \bibinfo{pages}{71--89},
  \doi{10.4204/EPTCS.45.6}.

\bibitemdeclare{book}{pier02}
\bibitem{pier02}
\bibinfo{author}{Benjamin~C. \surnamestart Pierce\surnameend}
  (\bibinfo{year}{2002}): \emph{\bibinfo{title}{Types and Programming
  Languages}}.
\newblock \bibinfo{publisher}{MIT Press}.

\bibitemdeclare{misc}{savara}
\bibitem{savara}
\emph{\bibinfo{title}{{{S}avara {JB}oss {P}roject}}}.
\newblock \bibinfo{note}{\url{http://www.jboss.org/savara}}.

\bibitemdeclare{inproceedings}{THK}
\bibitem{THK}
\bibinfo{author}{Kaku \surnamestart Takeuchi\surnameend},
  \bibinfo{author}{Kohei \surnamestart Honda\surnameend} \&
  \bibinfo{author}{Makoto \surnamestart Kubo\surnameend}
  (\bibinfo{year}{1994}): \emph{\bibinfo{title}{{An Interaction-based Language
  and its Typing System}}}.
\newblock In: {\sl \bibinfo{booktitle}{PARLE'94}}, {\sl \bibinfo{series}{LNCS}}
  \bibinfo{volume}{817}, pp. \bibinfo{pages}{398--413},
  \doi{10.1007/3-540-58184-7\_118}.

\bibitemdeclare{misc}{UNIFI}
\bibitem{UNIFI}
\bibinfo{author}{\surnamestart UNIFI\surnameend} (\bibinfo{year}{2002}):
  \emph{\bibinfo{title}{{I}nternational {O}rganization for {S}tandardization
  {ISO} 20022 {UNI}versal {F}inancial {I}ndustry message scheme}}.
\newblock \bibinfo{howpublished}{\url{http://www.iso20022.org}}.

\bibitemdeclare{inproceedings}{YDBH10}
\bibitem{YDBH10}
\bibinfo{author}{Nobuko \surnamestart Yoshida\surnameend},
  \bibinfo{author}{Pierre{-}Malo \surnamestart Deni{\'{e}}lou\surnameend},
  \bibinfo{author}{Andi \surnamestart Bejleri\surnameend} \&
  \bibinfo{author}{Raymond \surnamestart Hu\surnameend} (\bibinfo{year}{2010}):
  \emph{\bibinfo{title}{Parameterised Multiparty Session Types}}.
\newblock In: {\sl \bibinfo{booktitle}{FOSSACS}}, {\sl \bibinfo{series}{LNCS}}
  \bibinfo{volume}{6014}, \bibinfo{publisher}{Springer}, pp.
  \bibinfo{pages}{128--145}, \doi{10.1007/978-3-642-12032-9\_10}.

\end{thebibliography}

\end{document}